\documentclass[a4paper,10pt]{article}

\title{Search by a Metamorphic Robotic System in a Finite 3D Cubic Grid} 

\author{Ryonosuke Yamada\thanks{School of Information Science and Electrical Engineering, Kyushu University, Japan. E-mail: \texttt{yamada.ryonosuke.000@s.kyushu-u.ac.jp}} \and 
Yukiko Yamauchi\thanks{Faculty of Information Science and Electrical Engineering, Kyushu University, Japan. E-mail. \texttt{yamauchi@inf.kyushu-u.ac.jp}} 
}
\usepackage{longtable}
\usepackage{amsmath,amsfonts} 
\usepackage[pdftex]{graphicx}
\usepackage{latexsym}

\newtheorem{theorem}{Theorem}

\newenvironment{proof}{{\bf Proof. } } 

\newcommand{\qed}{\hfill $\Box$} 
\newcommand{\Order}{\mathrm{O}}

\begin{document}
\date{}
\maketitle

\begin{abstract}
We consider search in a finite 3D cubic grid by a \emph{metamorphic robotic system} (MRS), 
that consists of anonymous modules. 
A module can perform a sliding and rotation while the whole modules keep connectivity. 
As the number of modules increases, 
the variety of actions that the MRS can perform increases. 
The \emph{search problem} requires the MRS to find a target 
in a given finite field. 
Doi et al. (SSS 2018) demonstrate a necessary and sufficient number of modules for search 
in a finite 2D square grid. 
We consider search in a finite 3D cubic grid 
and investigate the effect of common knowledge. 
We consider three different settings. 
First, we show that three modules are necessary and sufficient when 
all modules are equipped with a \emph{common compass}, 
i.e., they agree on the direction and orientation of 
the $x$, $y$, and $z$ axes. 
Second, we show that four modules are necessary and sufficient when 
all modules agree on the direction and orientation of 
the vertical axis. 
Finally, we show that five modules are necessary and sufficient 
when all modules are not equipped with a common compass. 
Our results show that the shapes of the MRS in the 3D cubic grid have richer structure than 
those in the 2D square grid.

\end{abstract}

\section{Introduction}
\label{intro}

Swarm intelligence has shed light to collective behavior of 
autonomous entities with simple rules, 
such as ant, boid, and particles. 
The notion is applied to a collection of robots and 
swarm robotics has attracted much attention in the past two decades. 
Each autonomous element of the system is called 
a robot, module, agent, process, and sensor, 
and a variety of swarm robot systems have been investigated  
such as the \emph{autonomous mobile robot system}~\cite{SY99}, 
the \emph{population protocol model}~\cite{AADFP04}, 
the \emph{programmable particles}~\cite{DDGRSS14}, 
\emph{Kilobot}~\cite{RCN14}, 
and \emph{3D Catoms}~\cite{TPB19}. 
Dumitrescu et al. considered the \emph{metamorphic robotic system} (MRS), 
that consists of a collection of \emph{modules} in the infinite 2D square grid~\cite{DSY04a,DSY04b}. 
The modules are \emph{anonymous}, i.e., they are indistinguishable. 
They are \emph{autonomous} and \emph{uniform}, i.e., 
each module autonomously decides its movement by a common algorithm. 
They are \emph{oblivious}, i.e., each module has no memory of past. 
Thus, each module decides its behavior by 
observing other modules in nearby cells. 
Each module can perform a \emph{sliding} to a side-adjacent cell 
and a \emph{rotation} by $90$ degrees around a cell. 
The modules must keep \emph{connectivity}, which is defined by 
side-adjacency of cells occupied by modules. 
The authors considered \emph{reconfiguration} of the MRS, 
that requires the MRS to change the initial shape 
to a specified final shape~\cite{DSY04a}. 
They showed that any horizontally convex connected initial shape of an MRS 
can be transformed to any convex final shape via a straight chain shape. 
Later Dumitrescu and Pach showed that any connected initial shape 
can be transformed to any connected final shape via a 
straight chain shape~\cite{DP06}. 
In other words, the MRS has the ability of "universal reconfiguration." 

Reconfiguration can generate dynamic behavior of the MRS. 
Dumitrescu et al. demonstrated that the MRS can move forward 
by repeating a reconfiguration~\cite{DSY04b}. 
They showed a reconfiguration that realizes the fastest \emph{locomotion}. 
Doi et al. pointed out that the oblivious modules can use the 
shape of the MRS as global memory and 
the MRS can solve more complicated problem 
as the number of modules increases. 
They investigated \emph{search} in a finite 2D square grid, 
that requires the MRS to find a target cell in a finite 
rectangle \emph{field}~\cite{DYKY18}. 
Each module does not know the position of the target cell or   
the initial configuration 
of the MRS. 
They showed that the MRS can find the target cell 
by sweeping each row of the field without such a priori knowledge. 
Specifically, if the modules agree on the cardinal directions 
(i.e., north, south, east and west), 
three modules are necessary and sufficient, 
otherwise five modules are necessary and sufficient. 
Nakamura et al. considered \emph{evacuation} of the MRS from a 
finite rectangular field in the 2D square grid~\cite{NSY20}. 
There is a hole (i.e., two side-adjacent cells) 
on the wall of the field, 
and the MRS is required to exit from it from an arbitrary 
initial position and arbitrary initial shape. 
They showed that two modules are necessary and sufficient 
when the modules agree on the cardinal directions, 
otherwise four modules are necessary and sufficient. 

In this paper, we investigate the effect of common knowledge 
on search by the MRS in the 3D cubic grid. 
We consider the following three cases: 
$(i)$ modules equipped with a common \emph{compass} 
(i.e., they agree on the direction and orientation of $x$, $y$, and $z$ axes).
$(ii)$ modules equipped with a common vertical axes 
(i.e., they agree on the direction and orientation of $z$ axes), 
and 
$(iii)$ modules not equipped with a common compass 
(i.e., they have no agreement on directions and orientations). 
We demonstrate that 
three modules are necessary and sufficient when 
the modules are equipped with a common compass 
and five modules are necessary and sufficient 
when the modules are not equipped with a common compass. 
The numbers of sufficient modules in the 3D cubic grid 
are the same as those in the 2D square grid~\cite{DYKY18} 
because the MRS has more states in the 3D cubic grid 
than in the 2D square grid. 
For the intermediate case with a common vertical axis, 
we demonstrate that four modules are necessary and sufficient. 
Thus, our results in the 3D cubic grid 
show a smooth trade-off between 
the computational power of the MRS and 
common knowledge among modules, that the previous results 
in the 2D square grid could not find.  
We present search algorithms for these three settings 
and 
show the necessity by examining the 
state transition graph of the MRS. 

\noindent{\bf Related work.}~Reconfiguration of swarm robot systems 
have been discussed for the MRS~\cite{DP06,DSY04a}, 
autonomous mobile robots~\cite{FYOKY15,SY99} and 
the programmable particles~\cite{DGRSS15,DFSVY20}. 
Michail et al. considered the \emph{programmable matter system},
that is similar to the MRS
and investigated reconfiguration by rotations only and 
that by rotations and slidings~\cite{MSS19}. 
They showed that the combination of rotations and slidings 
guarantees universal reconfiguration, 
while rotations only cannot. 
They also presented $\Order(n^2)$-time reconfiguration algorithm 
by rotations and slidings. 
Almethen et al. considered reconfiguration by \emph{line-pushing}, 
where each module is equipped with the ability of pushing a line of 
modules~\cite{AMP20a}. 
They presented $\Order(n \log n)$-time universal reconfiguration algorithm 
that does not promise connectivity of intermediate shapes 
and $\Order(n \sqrt{n})$-time reconfiguration algorithm 
that transforms a diagonal line into a straight chain with 
preserving connectivity. 
The same authors later showed that their programmable matter system 
has the ability of universal reconfiguration 
and $\Order(n \sqrt{n})$-time reconfiguration algorithm 
together with $\Omega(n \log n)$ lower bound~\cite{AMP20b}. 

\section{Preliminary}
\label{junbi}

We consider a \emph{metamorphic robotic system} (MRS) in a finite 3D cubic grid. 
A metamorphic robotic system consists of 
a collection of anonymous (i.e., indistinguishable) ¥emph{modules}. 
A module can observe the positions of other modules in nearby cells, 
computes its next movement with a common algorithm, 
and performs the movement. 

Each cell of the 3D cubic grid can adopt at most one module at a time. 
Cell $(x,y,z)$ is the cell surrounded by grid points  
$(x, y, z)$, $(x+1, y, z)$, $(x, y+1, z)$, $(x, y, z+1)$, 
$(x+1, y+1, z)$, $(x+1, y, z+1)$, $(x, y+1, z+1)$, and $(x+1, y+1, z+1)$. 
Cells $(x+1, y, z)$, $(x, y+1, z)$, $(x, y, z+1)$, 
$(x-1, y, z)$, $(x, y-1, z)$, $(x, y, z-1)$ are \emph{side-adjacent} to 
cell $(x, y, z)$. 
We consider the positive $x$ direction as East, 
the positive $y$ direction as North, 
and the positive $z$ direction is Up. 

The MRS moves in a finite \emph{field}, 
which is a cuboid of width $w$, depth $d$ and height $h$ 
with its two diagonal cells being $(0,0,0)$ and $(w-1, d-1, h-1)$. 
We consider two types of \emph{planes}; 
the first type is a set of cells forming a plane 
perpendicular to one of the $x$, $y$, and $z$ axes. 
The second type is a set of cells parallel to one of the 
$x$, $y$, and $z$ axes and diagonal to the remaining two axes. 
For example, $\{(x,y,z) | y = s\} $ for some $s \in \mathcal{Z}$ is a 
plane of the first type 
and $\{(x, y, z) | x+y = s'\}$ for some $s' \in \mathcal{Z}$ is 
a plane of the second type. 
A \emph{line} of cells is a set of cells forming a horizontal or vertical line 
on a plane. 
For example, $\{(x,y,z) | y = u, z = v\} $
for some $u,v \in \mathcal{Z}$ is a line 
and $\{(x,y,z) | x+y = u', z = v'\}$ 
for some $u', v' \in \mathcal{Z}$ is a line. 

The field is surrounded by six planes, which we call \emph{walls}. 
More precisely, the walls are 
$\{(x, y, z) | x = -1\}$ (the West wall), 
$\{(x, y, z) | y = -1\}$ (the South wall), 
$\{(x, y, z) | z = -1\}$ (the Bottom wall), 
$\{(x, y, z) | x = w\}$ (the East wall), 
$\{(x, y, z) | y = d\}$ (the North wall), 
and $\{(x, y, z) | z = h\}$ (the Top wall). 

All modules synchronously perform observation, computation, and movement in 
each discrete time $t=0, 1, 2, \ldots$. 
A \emph{configuration} of the MRS is the set of cells occupied by the modules. 
We say two modules are \emph{side-adjacent} if they are in the two side-adjacent cells. 
We also say that a module $m$ is side-adjacent to cell $c$ if the cell occupied by $m$ is 
side-adjacent to $c$. 
Given a configuration of the MRS, consider a graph where 
each vertex corresponds to a module and there is an edge between two vertices 
if the corresponding modules are side-adjacent. 
If this graph is connected, we say the MRS is \emph{connected}. 

A module can perform two types of movements, \emph{sliding} and \emph{rotation}. 
\begin{enumerate} 
\item Sliding: 
When two modules $m_i$ and $m_j$ are side-adjacent, 
another module $m_k$ can move from a cell side-adjacent to $m_i$ 
to an empty cell side-adjacent to $m_k$ and $m_j$ along $m_i$ and $m_j$. 
During the movement, $m_i$ and $m_j$ cannot perform any movement. 
See Figure~\ref{Module-Move} as an example. 
\item Rotation: 
When two modules $m_i$ and $m_j$ are side-adjacent, 
$m_i$ can move to a cell side-adjacent to $m_j$ by rotating $\pi/2$ in some 
direction. 
There are six cells side-adjacent to $m_j$ and $m_i$ can move to four of them 
by rotation. 
During the movement, $m_j$ cannot move and the cells that $m_i$ passes must be empty. 
See Figure~\ref{Module-Move} as an example. 
\end{enumerate} 
Note that several modules can move at the same time as long as their moving tracks do not overlap. 
The modules must keep two types of connectivity at each time step. 
\begin{enumerate}
\item At the beginning of each time step, 
the modules must be connected. 
\item At each time step, the modules that do not move must be connected. 
\end{enumerate}

\begin{figure}
  \centering
  \includegraphics[keepaspectratio, scale=0.25]
       {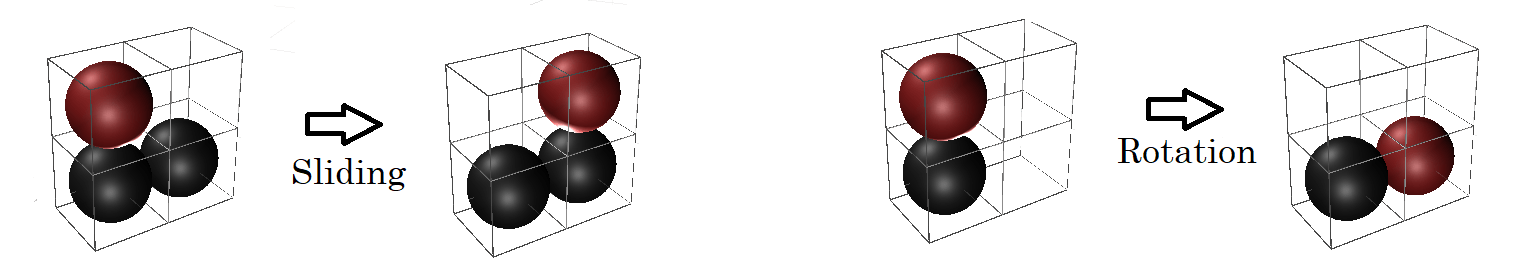}
  \caption{Sliding and rotation. The red modules perform movement.}
  \label{Module-Move}
 \end{figure}

We assume that each module obtains the result of an observation and moves to 
the next cell in its ¥emph{local $x$-$y$-$z$ coordinate system}. 
We assume that the origin of the local coordinate system of a module 
is its current cell and all local coordinate systems are right-handed. 
In this paper, we consider three types of MRSs with different degree of agreement 
on the coordinate system. 
When all modules agree on the directions and orientations of $x$, $y$, and $z$ axes, 
we say the MRS is equipped with a \emph{common compass}. 
When all modules do not agree on the directions or orientations of $x$, $y$, and $z$ axes, 
we say the MRS is not equipped with a common compass. 
Hence, local coordinate systems are not consistent among the modules. 
As an intermediate model, we consider modules that 
agree on the direction and orientation of the vertical axis. 
In this case, we say the MRS is equipped with a \emph{common vertical axis}. 
The \emph{state} of the MRS is its local shape. 
If the modules are equipped with a common compass or a common vertical axis, 
the state of the MRS contains common directions and orientations. 
Otherwise the state of the MRS does not contain any directions and orientations.

The \emph{search problem} requires the MRS to find a \emph{target} placed at one cell
in the field from a given initial configuration.
We call the cell containing the target the \emph{target cell}. 
The MRS \emph{finds} a target when one of its modules enters the target cell. 

\begin{figure}
  \centering
  \includegraphics[keepaspectratio, scale=0.35]
        {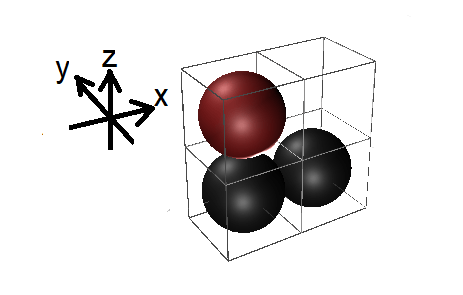}
  \caption{Example of an observation at the red module with the local coordinate system. }
  \label{Modules}
\end{figure}

When a module executes the common algorithm, 
the input is the \emph{observation} of cells in a cube of size
$(2k+1) \times (2k+1) \times (2k+1)$ centered at the module 
(i.e., its \emph{$k$-neighborhood}).
It detects whether each cell in its $k$-neighborhood 
is a wall cell or not and
whether it is occupied by a module or not.
Let $C_m$ be the set of cells occupied by some modules,
$C_w$ be the set of wall cells,
and $C_e$ be the set of the remaining (i.e., empty) cells
of an observation.
More precisely, each set of cells is a set of coordinates of the
corresponding cells observed in the local coordinate system of the module.
For example, in Figure~\ref{Modules}, 
the result of an observation at the red module is
$C_m = \{(0, -1, 0), (1, -1, 0)\}$, $C_w = \emptyset$, and
$C_e $ is the remaining cells.
When a common algorithm outputs coordinate $(a,b,c)$ at a module, 
the module moves to $(a, b, c)$ in its local coordinate system.

When we describe an algorithm, 
the elements of $C_m$, $C_w$, and $C_e$ are specified in a "canonical coordinate system," i.e., 
the global coordinate system. 
When the modules are equipped with a common compass, without loss of generality, 
we assume that the common compass is identical to the global coordinate system. 
Thus, each module obtains its movement by checking 
$C_m$, $C_w$, and $C_e$. 
When the modules agree on a common vertical axis, 
without loss of generality, we assume that 
the vertical axis is identical to the $z$ axis of the global coordinate system. 
Each module computes its movement by 
rotating the current observation by $\pi/2$, $\pi$, $3\pi/2$, and $2\pi$ around 
the common vertical axis and comparing the results with 
$C_m$, $C_w$, and $C_e$. 
It selects the output with movement 
and if there are multiple outputs with movement 
it nondeterministically selects one of them. 
When the modules are not equipped with a common compass, 
a module checks $24$ rotations of the current observation and 
selects an output in the same way as the above case. 

\section{Search algorithms for MRSs in a finite 3D cubic grid} 

In this section, we present search algorithms with small number of modules. 
Our proposed algorithms are based on a common strategy. 
Since the MRS does not know the position of the target cell, 
we make the MRS visit all cells of the field. 
The proposed algorithms slice the field into planes 
and the MRS visits the cells of each plane 
by sweeping each row or column of the plane. 
Thus, the proposed algorithms are extensions of the search algorithms 
by Doi et al. in a finite 2D cubic grid~\cite{DYKY18}. 

\subsection{Search with a common compass}
\label{ExpWithCompass}

We show the following theorem by a search algorithm for the MRS of 
three modules equipped with a common compass. 
\begin{theorem}
\label{theorem: 3 modules can Exprole 3D space with compass}
The MRS of three modules equipped with a common compass can solve the search problem 
in a finite 3D cubic grid from any initial configuration.
\end{theorem}

The proposed algorithm considers each plane 
$\{(x,y,z) | x+y=s\}$ for $s = 1, 2, \ldots$ and 
the MRS moves along each line $\{(x,y,z) | x+y=s, z=t\} $ 
for $t= 0, 1, 2, \ldots$. 
Figure~\ref{Explode-Move3} shows a moving track of the proposed aglorithm.  
The MRS continues to search each planes until it reaches the northeasternmost plane. 
Then, it moves along the edges of the field so that it returns to the 
southwesternmost plane. 
It starts searching each plane again to visit all cells of the field.

The MRS moves forward or turns by repeating a sequence of movements, that we call a \emph{move sequence}. 
The proposed algorithm consists of the following move sequences. 
%

\begin{itemize}
     \item Move sequence $M_{NW}$ (Figure~\ref{3MoveToNorthwest}). 
     The blue module is in cell $(x,y,z)$ at the first.
     By this move sequence, 
     the green module reaches cell $(x-1,y+1,z)$.
     By repeating $M_{NW}$ $n$ times, any of the modules visit the cells $(x-k,y+k,z)(0 \leq k \leq n)$. 
     That is, it visits all the cells of the horizontal line $\{(x,y,z)|x+y=s,z=t\}$. 
     
     \item Move sequence $M_{TurnNW}$ (Figure~\ref{3TurnOnTheNorthWestWall}). 
     By this move sequence, 
     the MRS changes its move sequence from $M_{NW}$ to $M_{SE}$.
     
     \item Move sequence $M_{SE}$ (Figure~\ref{3MoveToSoutheast}). 
     The blue module is in cell $(x,y,z)$ at the first.
     By this move sequence, 
     the green module reaches cell $(x+1,y-1,z)$.
     By repeating $M_{SE}$ $n$ times, any of the modules visit the cells $(x+k,y-k,z)(0 \leq k \leq n)$.
     That is, it visits all the cells of the horizontal line $\{(x,y,z)|x+y=s,z=t\}$.
     
     \item Move sequence $M_{TurnSE}$ (Figure~\ref{3TurnOnTheSouthEastWall}). 
     By this move sequence, 
     the MRS changes its move sequence from $M_{SE}$ to $M_{NW}$.

     \item Move sequence $M_{T}$ (Figure~\ref{3MoveOnTheTopOfWall}). 
     By this move sequence, 
     the MRS changes its move sequence from $M_{SE}$ to $M_{D}$.

     \item Move sequence $M_{D}$ (Figure~\ref{3DecentOnTheEastWall}). 
    The blue module is in cell $(x,y,z)$ at the first.
    By this move sequence, 
    the green module reaches cell$(x,y,z-1)$.
     By repeating $M_{D}$ $n$ times, any of the modules visit the cells $(x,y,z-k)(0 \leq k \leq n)$.
     That is, it visits all the cells of the line $\{(x,y,z)|x=s,y=t\}$.

     \item Move sequence $M_{B}$ (Figure~\ref{3LeaveTheBottomOfTheEastWall}).
     By this move sequence,
     the MRS changes its move sequence from $M_{D}$ to $M_{NW}$.

     \item Move sequence $M_{NECorner}$ (Figure~\ref{3MoveOnTheNortheastCorner}). 
     By this move sequence,
     the MRS changes its move sequence from $M_{D}$ to $M_{WallBottom}$.

     \item Move sequence $M_{WallBottom}$ (Figure~\ref{3MoveAlongTheBottomOfTheWall}). 
    The blue module is in cell $(x,y,0)$ at the first.
    By this move sequence, 
    the green module reaches cell $(x,y-1,0)$ along the edge.
     By repeating $M_{WallBottom}$ $n$ times, any of the modules visit the cells $(x,y-k,0)(0 \leq k \leq n)$.
     That is, it visits all the cells of the line $\{(x,y,z)|x=s,z=0\}$.

     \item Move sequence $M_{SWCorner}$ (Figure~\ref{3MoveOnTheSouthwestCorner}). 
     By this move sequence,
     the MRS changes its move sequence from $M_{WallBottom}$ to $M_{Up}$.

     \item Move sequence $M_{Up}$. (Figure~\ref{3RiseOnTheSouthwestCorner})
    The blue module is in cell $(0,0,z)$ at the first.
    By this move sequence, 
    the green module reaches cell $(0,0,z+1)$.
     By repeating $M_{Up}$ $n$ times, any of the modules visit the cells $(0,0,k)(0 \leq k \leq n)$.
     That is, it visits all the cells of the line $\{(x,y,z)|x=0,y=0\}$.
\end{itemize}

    \begin{figure}[ht]
      \centering
      \includegraphics[keepaspectratio, scale=0.3]
           {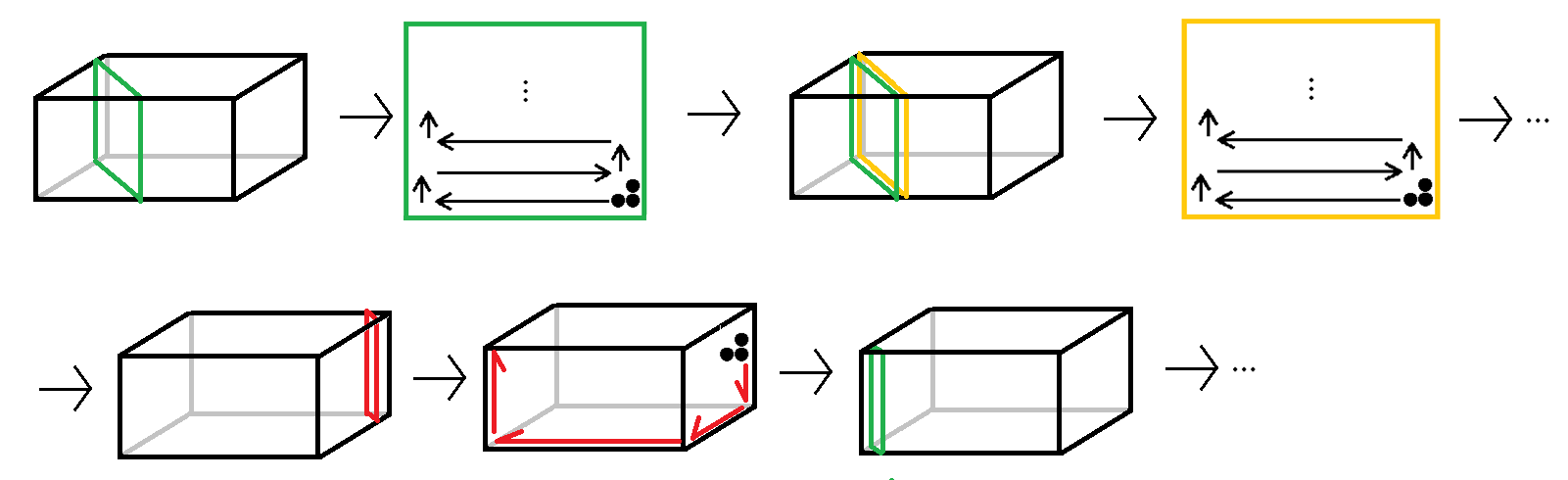}
      \caption{Search by three modules equipped with common compass}
      \label{Explode-Move3}
     \end{figure}

     \begin{figure}[p]
      \centering
      \includegraphics[keepaspectratio,  height=2.3cm]
           {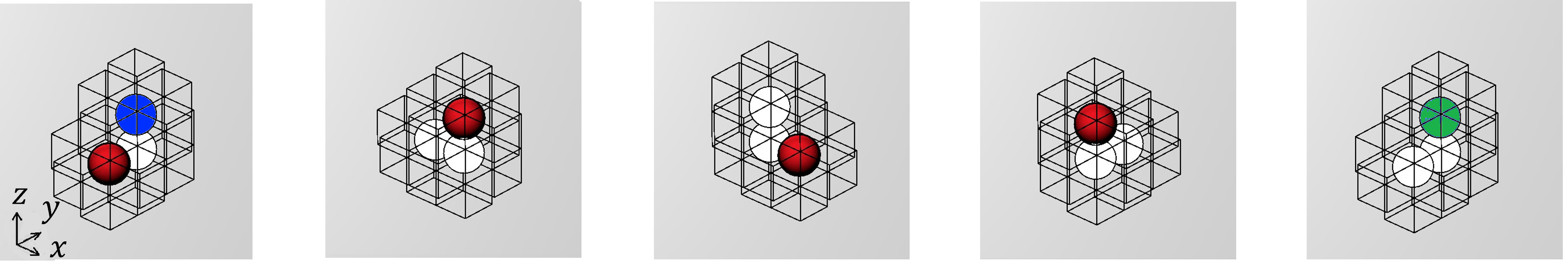}
      \caption{Move to northwest. In each figure, the red module moves. When the blue module is in cell $(x,y,z)$ at the start, after this move sequence, the green model reaches cell $(x-1, y+1, z)$. } 
      \label{3MoveToNorthwest}
     \end{figure}

     \begin{figure}[p]
      \centering
      \includegraphics[keepaspectratio, height=2.5cm]
            {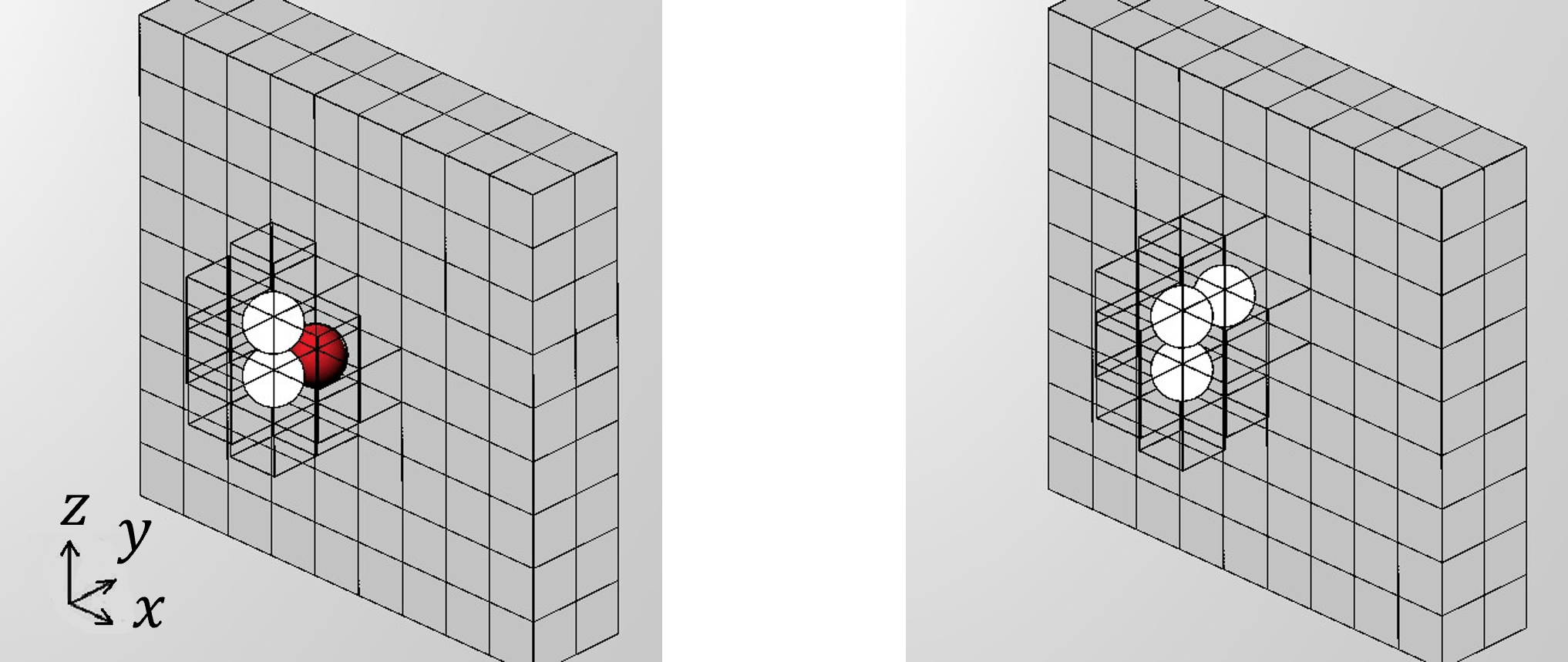}
      \caption{Turn on the north or west wall. In the first figue, the red module moves. }
      \label{3TurnOnTheNorthWestWall}
     \end{figure}

     \begin{figure}[p]
      \centering
      \includegraphics[keepaspectratio, height=2.3cm]
           {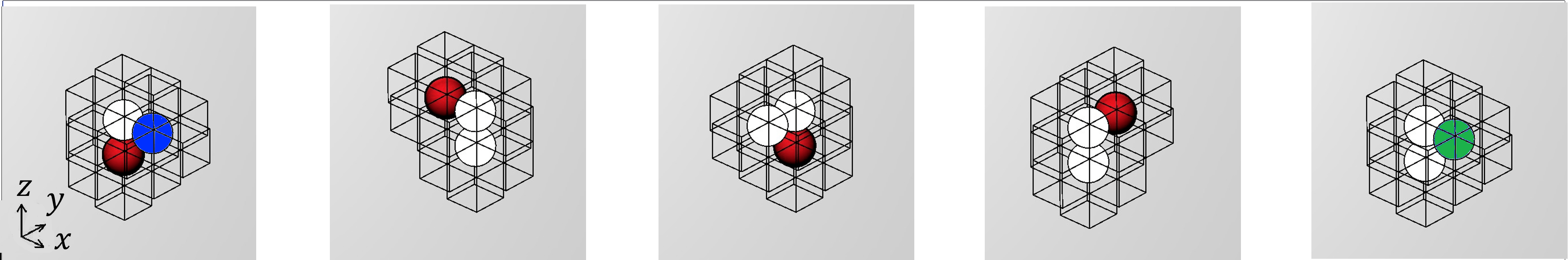}
      \caption{Move to southeast. In each figure, the red module moves. When the blue module is in cell $(x,y,z)$ at the start, after this move sequence, the green model reaches
      $(x+1,y-1,z)$. }
      \label{3MoveToSoutheast}
     \end{figure}

     \begin{figure}[p]
      \centering
      \includegraphics[keepaspectratio, height=2.5cm]
           {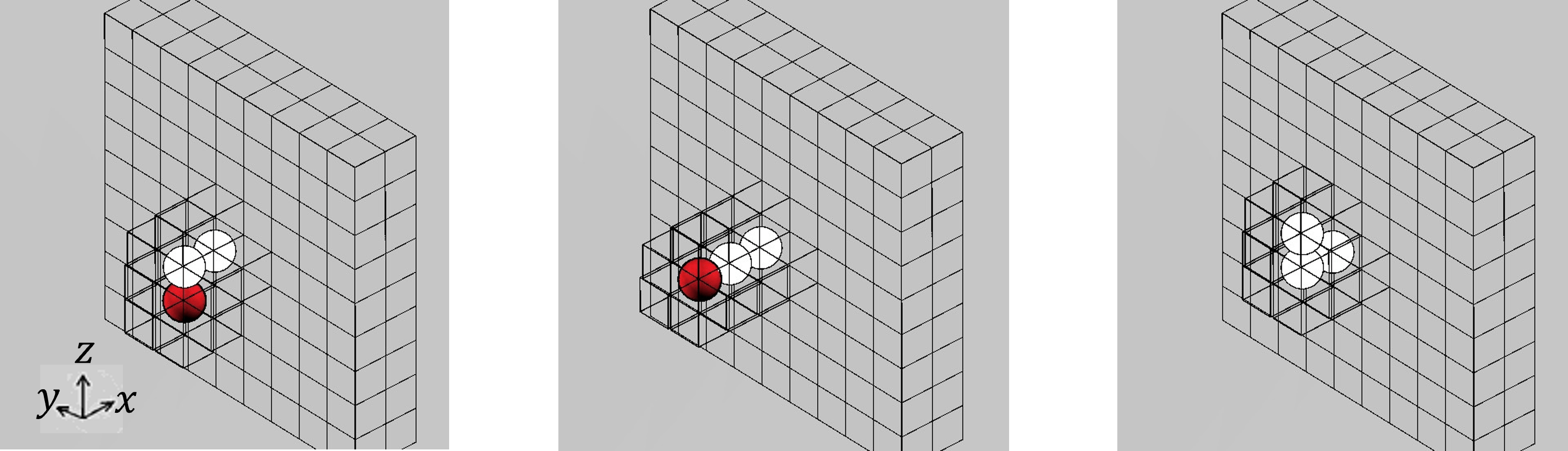}
      \caption{Turn on the south or east wall. In each figure, the red module moves.}
      \label{3TurnOnTheSouthEastWall}
     \end{figure}

     \begin{figure}[p]
      \centering
      \includegraphics[keepaspectratio, height=2.5cm]
           {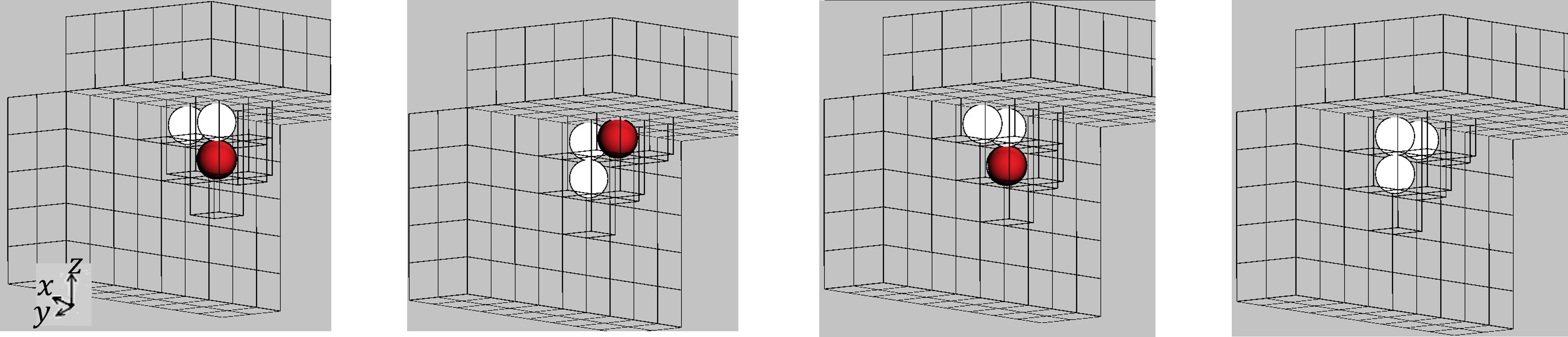}
      \caption{Move around the top of the north or east wall. In each figure, the red module moves.}
      \label{3MoveOnTheTopOfWall}
     \end{figure}

     \begin{figure}[p]
      \centering
      \includegraphics[keepaspectratio, height=2.5cm]
           {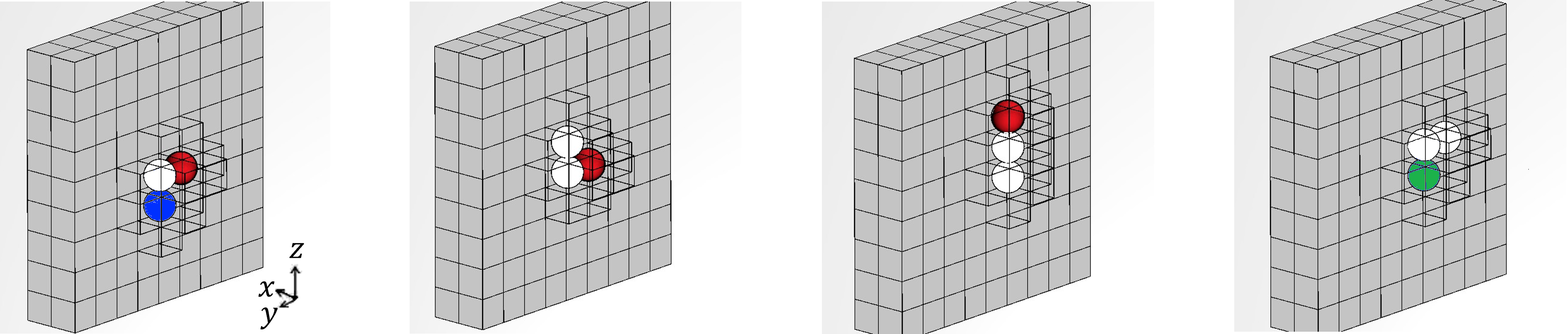}
      \caption{Move down on the north or east wall.
      In each figure, the red module moves. When the blue module is in cell $(x,y,z)$ at the start, after this move sequence, the green model reaches cell $(x, y, z-1)$.}
      \label{3DecentOnTheEastWall}
     \end{figure}

     \begin{figure}[p]
      \centering
      \includegraphics[keepaspectratio, height=2.5cm]
           {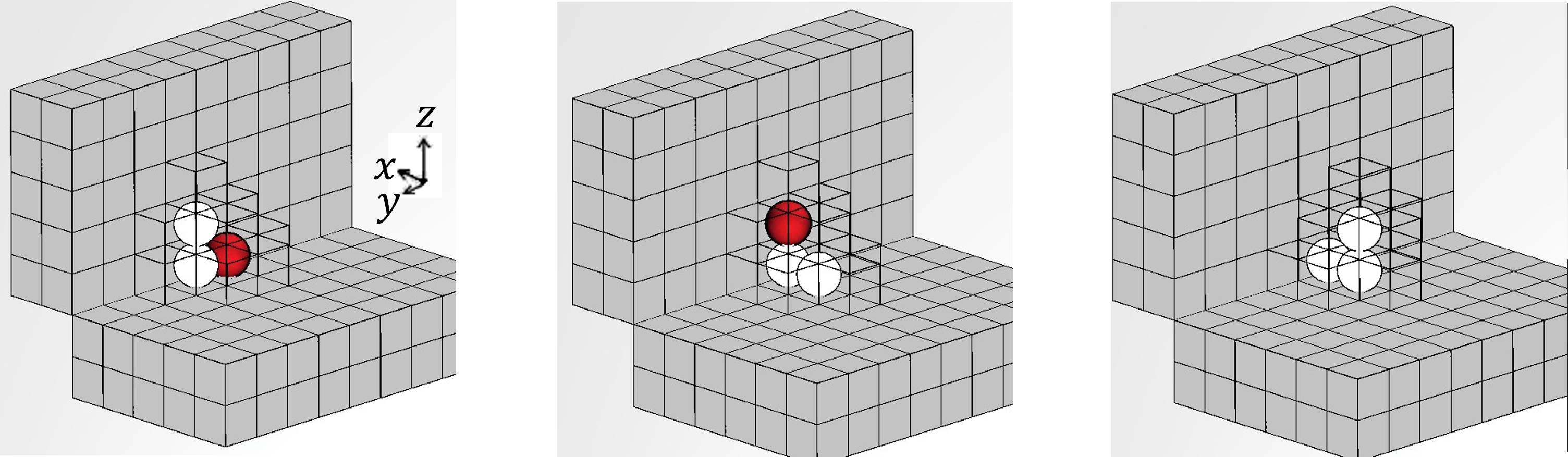}
      \caption{Leaving the bottom of the north or east wall. In each figure, the red module moves.}
      \label{3LeaveTheBottomOfTheEastWall}
     \end{figure}

     \begin{figure}[p]
      \centering
      \includegraphics[keepaspectratio, height=2.5cm]
           {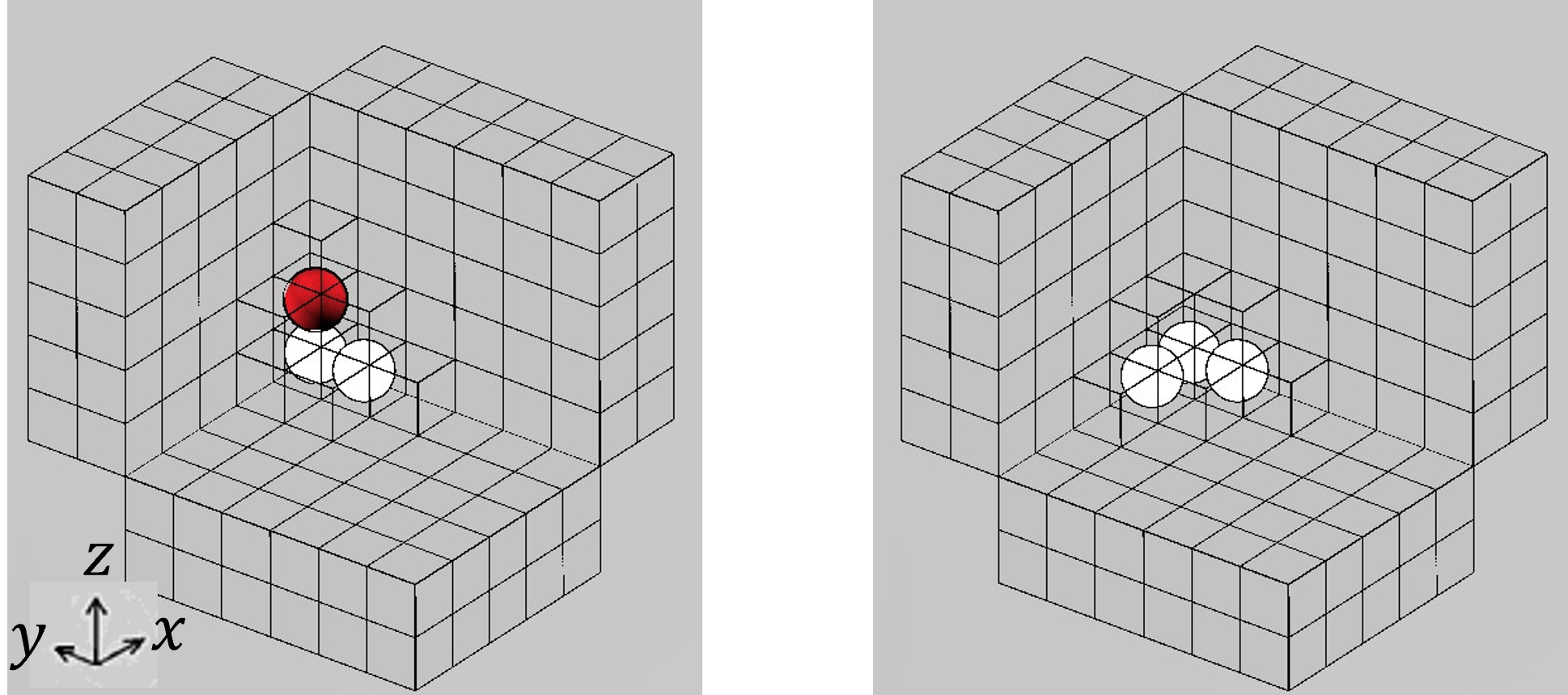}
      \caption{Move on the northeast corner. In the first figure, the red module moves.}
      \label{3MoveOnTheNortheastCorner}
     \end{figure}

     \begin{figure}[hp]
      \centering
      \includegraphics[keepaspectratio, height=2.5cm]
           {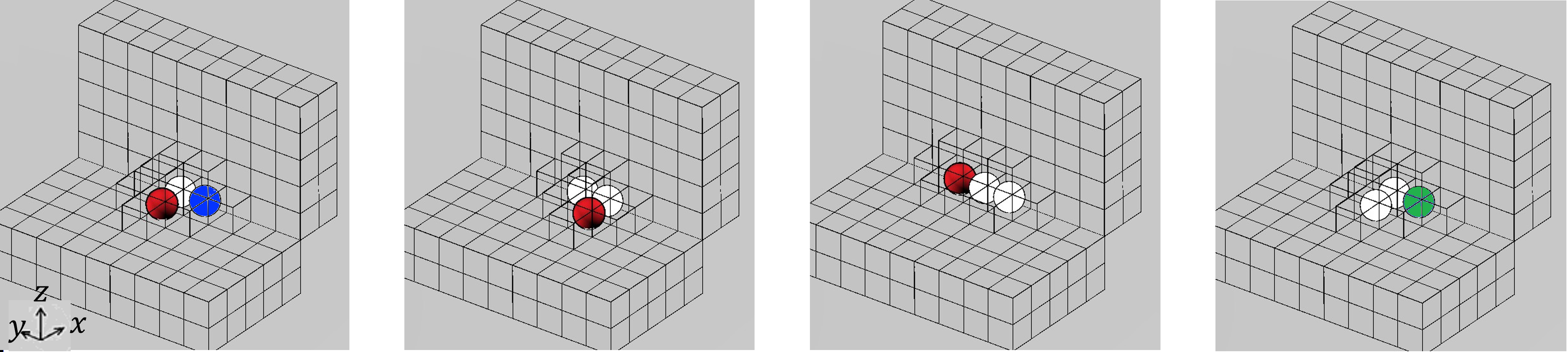}
      \caption{Move along the bottom of the wall.
       In each figure, the red module moves. When the blue module is in cell $(x,y,0)$ at the start, after this move sequence, the green model reaches cell $(x, y-1, 0)$.}
      \label{3MoveAlongTheBottomOfTheWall}
     \end{figure}

     \begin{figure}[hp]
      \centering
      \includegraphics[keepaspectratio, height=2.5cm]
           {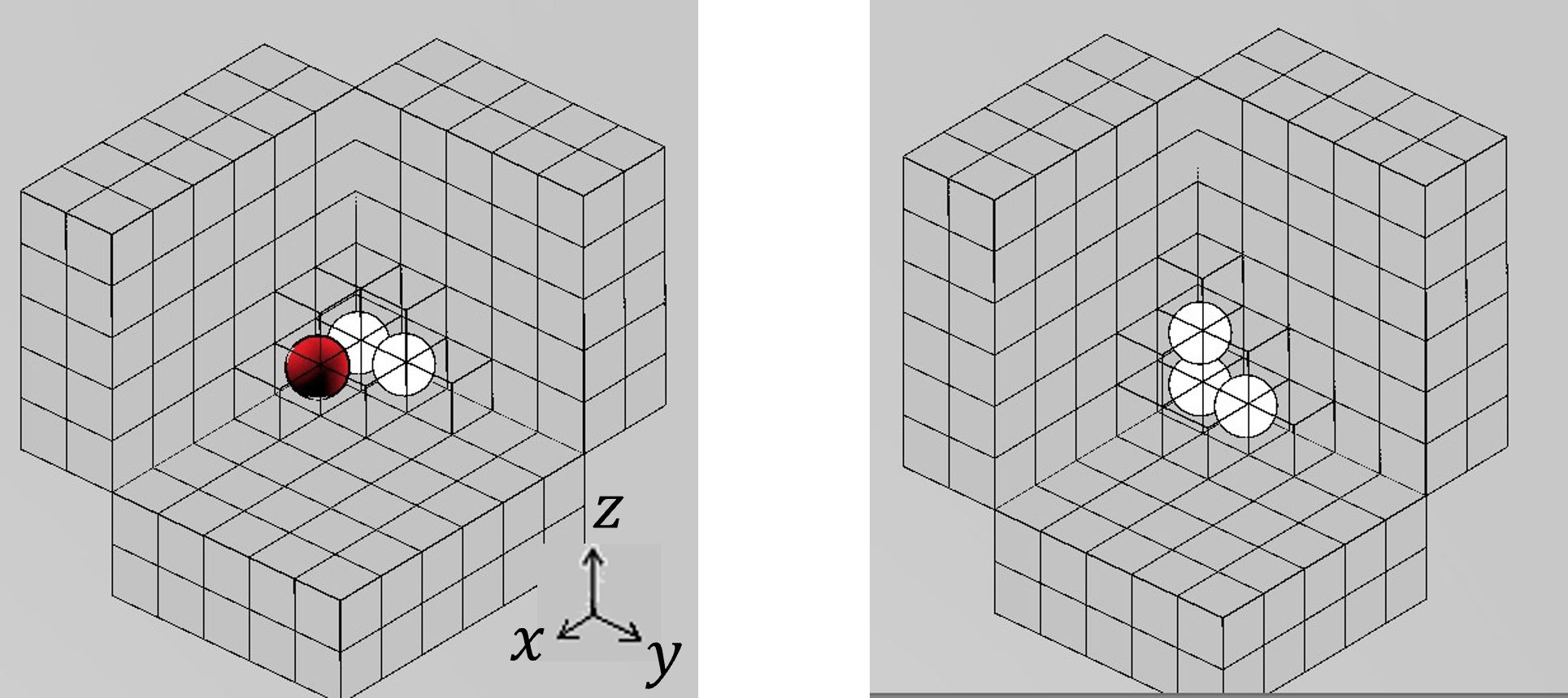}
      \caption{Move on the southwest corner. In the first figure, the red module moves.}
      \label{3MoveOnTheSouthwestCorner}
     \end{figure}
     \begin{figure}[hp]
      \centering
      \includegraphics[keepaspectratio, height=2.5cm]
           {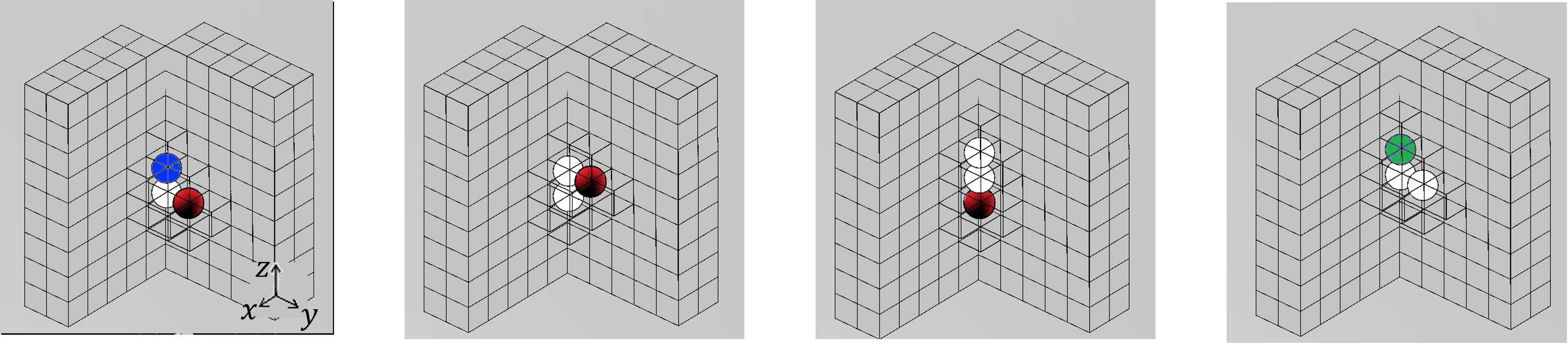}
      \caption{Move up on the southwest corner. In each figure, the red module moves. When the blue module is in cell $(0,0,z)$ at the start, after this move sequence, the green model reaches cell $(0, 0, z+1)$.}
      \label{3RiseOnTheSouthwestCorner}
     \end{figure}


The proposed algorithm consists of the following seven steps. 
\begin{description}
\item[Step 1] The MRS repeats $M_{NW}$, 
    that makes it move in the northwest direction 
    along a horizontal line on a plane $\{(x,y,z) | x + y = s\}$ for some $s$. 
\item[Step 2] When the MRS reaches the north or west wall, 
    it changes the moving direction to southeast 
    by $M_{TurnNW}$. 
\item[Step 3] The MRS repeats $M_{SE}$, 
    that makes it move in the southeast direction 
    along a horizontal line on $\{(x,y,z) | x + y = s\}$. 
    This movement makes the MRS move along the same horizontal line as Step $1$. 
\item[Step 4] If the MRS is adjacent to the top wall 
    it moves to the plane $\{(x,y,z) | x + y = s+1\}$ 
    by $M_{T}$. 
    Then, it repeats $M_{D}$ shown in Figure \ref{3DecentOnTheEastWall}, 
    that makes it move down along the south wall or east wall 
    until it reaches the bottom wall. 
    Then, it leaves the wall by $M_{B}$. 
    It starts searching the new plane by repeating Steps $1$, $2$, $3$, and $4$. 
    Otherwise, it proceeds to Step $5$. 
\item[Step 5] When the MRS reaches the south or east wall, 
    it moves to the row above by $M_{TurnSE}$. 
    Then, it repeats Steps $1$, $2$, and $3$ so that it visits all cells on the new horizontal line. 
\item[Step 6] When the MRS reaches the northeast corner of the top wall, 
    the algorithm sends the MRS back to the southwest corner, 
    where the MRS starts searching by repeating Steps $1$ to $5$. 
    It moves along the northeast edge until it reaches the northeast corner 
    of the bottom wall by $M_D$.  
    Then, it moves along the east edge of the bottom wall 
    until it reaches the south east corner of the bottom wall 
    by $M_{NECorner}$ and repeating $M_{WallBottom}$. 
    It moves along the south edge of the bottom wall 
    until it reaches the southwest corner of the bottom wall 
    by repeating $M_{WallBottom}$. 
    Finally, it moves along the southeast edge 
    until it reaches the southwest corner of the top wall 
    by $M_{SWCorner}$ and repeating $M_{Up}$. 
    Then, the MRS returns to Step 4. 
\end{description}

Table~\ref{tab_3module} shows the input and the output of the proposed algorithm. 
Each element specifies a part of the input (especially, $C_w$ and $C_e$), 
and the MRS does not care whether other cells than those specified are walls or not. 

When the MRS is on a plane $\{(x, y, z) | x+y=s\}$ for some $s$, 
it visits all cells in the horizontal line $\{(x,y,z)|x+y=s,z=t\}$ for some $t$ 
by repeating Steps $1$, $2$, and $3$. 
Then, it proceeds to the horizontal line  $\{(x,y,z)|x+y=s,z=t+1\}$ by Step $5$. 
By repeating Steps $1$, $2$, $3$, and $5$, 
it eventually reaches the top wall. 
Then, it starts searching for cells in $\{(x,y,z)|x+y=s+1\}$ by Step $4$ and $6$. 

Repeating the above movement, the MRS eventually reaches the northeast corner of the top wall. 
At this point, it may have not yet visited the cells near the south west corner. 
Steps $6$ enables the MRS visit these cells 
by moving it to the southwest corner of the top wall and 
starting Steps $1$ again.

    \begin{figure}
      \centering
      \includegraphics[keepaspectratio, scale=0.35]
           {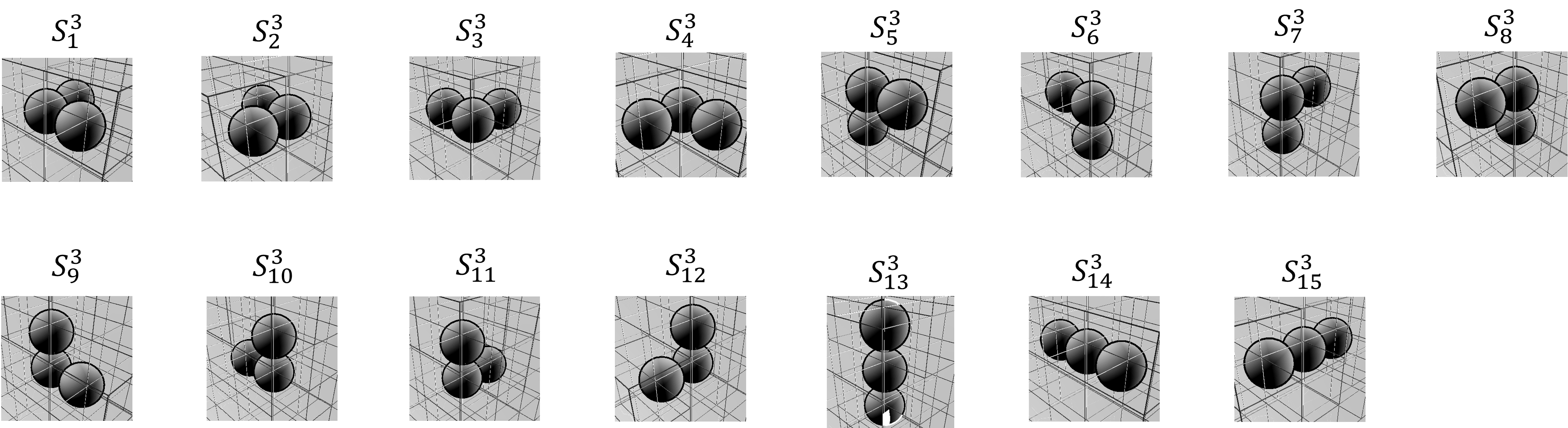}
      \caption{States of the MRS of three modules equipped with a common compass}
      \label{3ModulesShapes}
     \end{figure}

There exist initial configurations that satisfies no condition of Table~\ref{tab_3module}. 
We add exceptional transformation rules from such initial configurations. 
Figure~\ref{3ModulesShapes} shows all states of three modules equipped with a common compass.
Observe that any configuration can be transformed to another one in one time step. 
(Note that more than one module can move in one time.) 
Hence, even if the initial state of the MRS does not match any entry of Table~\ref{tab_3module}, 
the MRS can be transformed to one of the entries 
and the MRS can start search from any initial configuration.


\begin{longtable}{|l|c|c|c|c|}
  \caption{Search algorithm for the MRS of three modules equipped with a common compass}
  \endhead
  \endfoot
\cline{1-5}          & $C_m$ & $C_w$ & $C_e$ &Output \\
\cline{1-5}   $M_{SE}$ & $(0,0,1)$,$(1,0,1)$ &           &  $(2,0,0)$      & $(1,0,0)$ \\
          \cline{2-5} & $(1,0,0)$,$(1,0,-1)$ &           &            & $(1,-1,0)$ \\
       \cline{2-5}   & $(0,0,1)$,$(0,-1,1)$ &            &  $(0,-2,0)$        & $(0,-1,0)$ \\
      \cline{2-5}    & $(0,-1,0)$,$(0,-1,-1)$ &                  &       & $(1,-1,0)$ \\
\cline{1-5}    $M_{TurnSE}$ & $(0,0,1)$,$(1,0,1)$ & $(2,0,0)$      &  $(0,-1,0)$       & $(-1,0,1)$ \\
      \cline{2-5}    & $(1,0,0)$,$(2,0,0)$ &             &  $(0,0,-1)$, & $(1,0,1)$ \\
                   &  &             &  $(0,0,1)$ &\\
      \cline{2-5} & $(0,0,1)$,$(0,-1,1)$ & $(0,-2,0)$ &       & $(0,1,1)$ \\
      \cline{2-5}   & $(0,-1,0)$,$(0,-2,0)$ & $(0,0,-1)$ &            & $(0,-1,1)$ \\

\cline{1-5}  $M_{NW}$ & $(-1,0,0)$,$(-1,0,1)$ &          &  $(0,-1,0)$       & $(-1,1,0)$ \\
     \cline{2-5}     & $(0,0,-1)$,$(0,1,-1)$ &          &  $(0,2,0)$        & $(0,1,0)$ \\
      \cline{2-5}    & $(0,1,0)$,$(0,1,1)$ &          &  $(1,0,0)$      & $(-1,1,0)$ \\
       \cline{2-5}   & $(0,0,-1)$,$(-1,0,-1)$ &           &  $(-2,0,0)$        & $(-1,0,0)$ \\
\cline{1-5}    $M_{TurnNW}$ & $(0,-1,0)$,$(0,-1,1)$ & $(0,1,0)$      &             & $(0,0,1)$ \\
\cline{2-5}  & $(1,0,0)$,$(1,0,1)$ & $(-1,0,0)$ &         & $(0,0,1)$ \\
\cline{1-5}    $M_{T}$  & $(0,0,1)$,$(1,0,1)$ & $(2,0,0)$,$(0,0,2)$ &          & $(1,0,0)$ \\
       \cline{2-5}  & $(1,0,0)$,$(1,0,-1)$ & $(2,0,0)$,$(0,0,1)$&           & $(1,1,0)$ \\
        \cline{2-5}  & $(0,0,1)$,$(0,1,1)$ & $(1,0,0)$,$(0,0,2)$ &           & $(0,1,0)$ \\
        \cline{2-5}   & $(0,0,1)$,$(0,-1,1)$ & $(0,0,2)$,$(0,-2,0)$      &       & $(0,-1,0)$ \\
        \cline{2-5}  & $(0,-1,0)$,$(0,-1,-1)$ & $(0,0,1)$,$(0,-2,0)$ &       & $(1,-1,0)$ \\
        \cline{2-5}  & $(0,0,1)$,$(1,0,1)$ & $(0,0,1)$,$(0,-1,0)$ &          & $(1,0,0)$ \\

\cline{1-5}    $M_{D}$ & $(0,1,0)$,$(0,1,-1)$ & $(1,0,0)$  &        & $(0,0,-1)$ \\
       \cline{2-5}   & $(0,1,0)$,$(0,1,1)$ & $(1,0,0)$ &       $(0,0,-1)$     & $(0,1,-1)$ \\
       \cline{2-5}   & $(0,0,-1)$,$(0,0,-2)$ & $(1,0,0)$ &                & $(0,-1,-1)$ \\
\cline{2-5} & $(1,0,0)$,$(1,0,-1)$ & $(0,-1,0)$ &         & $(0,0,-1)$ \\
\cline{2-5}   & $(1,0,0)$,$(1,0,1)$ & $(0,-1,0)$ &      $(0,0,-1)$        & $(1,0,-1)$ \\
 \cline{2-5}  & $(0,0,-1)$,$(0,0,-2)$ & $(0,-1,0)$ &      & $(-1,0,-1)$ \\

\cline{1-5}  $M_{B}$  & $(0,1,0)$,$(0,1,1)$ & $(1,0,0)$,$(0,0,-1)$ &      $(0,-1,0)$, & $(-1,1,0)$ \\
                 & & &      $(0,2,0)$ & \\
       \cline{2-5}  & $(0,0,-1)$,$(-1,0,-1)$ & $(1,0,0)$,$(0,0,-2)$ &       $(0,-1,0)$   & $(-1,0,0)$ \\
        
 \cline{2-5}  & $(1,0,0)$,$(1,0,1)$ & $(0,-1,0)$,$(0,0,-1)$ &          & $(1,1,0)$ \\
 \cline{2-5}  & $(0,0,-1)$,$(0,1,-1)$ & $(0,-1,0)$,$(0,0,-2)$ &           & $(0,1,0)$ \\
\cline{1-5}   $M_{NECorner}$  & $(0,0,-1)$,$(0,-1,-1)$ & $(0,1,0)$,$(1,0,0)$,&    & $(-1,0,-1)$ \\
                             &                   &       $(0,0,-2)$ &    & \\
       \cline{1-5} $M_{WallBottom}$   & $(1,0,0)$,$(1,-1,0)$ & $(0,0,-1)$ &          & $(0,-1,0)$ \\
       \cline{2-5}   & $(1,0,0)$,$(1,1,0)$ & $(0,0,-1)$ &        $(0,-1,0)$     & $(1,-1,0)$ \\
       \cline{2-5}   & $(0,-1,0)$,$(0,-2,0)$ & $(0,0,-1)$ &          & $(-1,-1,0)$ \\
\cline{2-5}  & $(0,-1,0)$,$(-1,-1,0)$ & $(0,0,-1)$,$(0,-2,0)$  &       & $(-1,0,0)$ \\
       \cline{2-5}   & $(0,-1,0)$,$(1,-1,0)$ & $(0,0,-1)$ &      $(-1,0,0)$  & $(-1,-1,0)$ \\
       \cline{2-5}   & $(-1,0,0)$,$(-2,0,0)$ & $(0,0,-1)$ &         & $(-1,1,0)$ \\
\cline{1-5} $M_{SWCorner}$ & $(-1,0,0)$,$(-1,1,0)$ & $(0,0,-1)$,$(-2,0,0)$,&      & $(-1,0,1)$ \\
                             &                       & $(0,-1,0)$&      &  \\
        \cline{1-5} $M_{Up}$  & $(0,-1,0)$,$(0,-1,1)$ & $(-1,0,0)$,$(0,-2,0)$  &       & $(0,0,1)$ \\
        \cline{2-5}  & $(0,-1,0)$,$(0,-1,-1)$ & $(-1,0,0)$,$(0,-2,0)$ &     $(0,0,1)$    & $(0,-1,1)$ \\
        \cline{2-5}  & $(0,0,1)$,$(0,0,2)$ & $(-1,0,0)$,$(0,-1,0)$ &            & $(0,1,1)$ \\

          \cline{1-5}          & \multicolumn{1}{c|}{Otherwise} & \multicolumn{1}{c|}{Otherwise} & \multicolumn{1}{c|}{Otherwise} & $(0,0,0)$ \\

          \cline{1-5}   
          \end{longtable}
  \label{tab_3module}%

\subsection{Search with a common vertical axis} 

We show the following theorem by a search algorithm for the MRS of 
four modules equipped with a common vertical axis. 
  \begin{theorem}
    \label{theorem: 4 modules can Exprole 3D space without horizontal compass}
    The MRS of four modules with a common vertical axis can solve a search problem 
    in a finite 3D grid 
    if no pair of modules have an identical observation in an initial configuration. 
    \end{theorem}

We prove Theorem~\ref{theorem: 4 modules can Exprole 3D space without horizontal compass} 
by a search algorithm. 
The proposed algorithm considers each plane 
$x=s$ for $s = 0, 1, 2, \ldots$ and 
$y=s'$ for $s' = 0, 1, 2, \ldots$. 
The MRS moves along each vertical line $x=s,y=t$ when it is on the plane $x=s$ for $t=0,1,2,\ldots$ 
and $x=t', y=s'$  when it is on the plane $y=s'$ for $t'=0,1,2,\ldots$. 
Figure~\ref{Explode-Move4} shows an execution of the proposed algorithm. 

The MRS continues to search each plane perpendicular to the $x$-axis 
until it reaches the east wall. 
Then, it starts to search each plane perpendicular to the $y$-axis. 
The algorithm description is given in Table~\ref{tab_4module}.
 \begin{figure}[htbp]
      \centering
      \includegraphics[keepaspectratio, width=\linewidth]
           {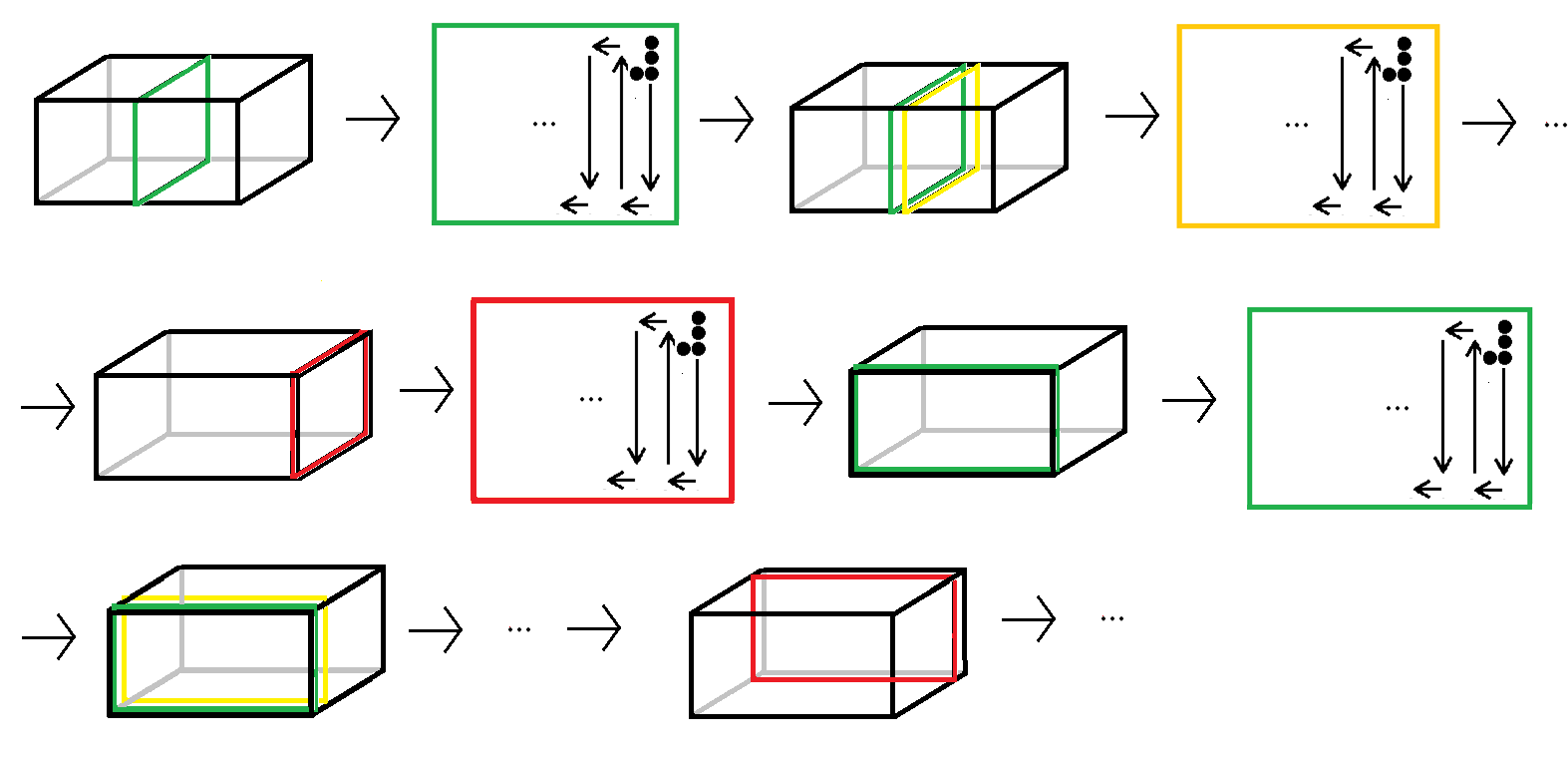}
      \caption{Example of a search by four modules}
      \label{Explode-Move4}
     \end{figure}

     The proposed algorithm consists of the following move sequences.
     \begin{itemize}
          \item Move sequence $M_{Down}$(Figure~\ref{4MoveToDown}).
          The blue module is in cell $(x,y,z)$ at the start.
          By this move sequence, 
          the green module reaches cell $(x,y,z-1)$.
          By repeating $M_{Down}$ $n$ times, any of the modules visit the cells $(x,y,z-k)(0 \leq k \leq n)$. 
          That is, it visits all the cells of the horizontal line $\{(x,y,z)|x=s,y=t\}$. 
     
          \item Move sequence $M_{Up}$(Figure~\ref{4MoveToUp}).
          The blue module is in cell $(x,y,z)$ at the start.
          By this move sequence, 
          the green module reaches cell $(x,y,z+1)$.
          By repeating $M_{Up}$ $n$ times, any of the modules visit the cells $(x,y,z+k)(0 \leq k \leq n)$. 
          That is, it visits all the cells of the horizontal line $\{(x,y,z)|x=s,y=t\}$. 
     
          \item Move sequence $M_{TurnUp}$(Figure~\ref{4TurnOnTheUpWall}).
          By this move sequence, 
          the MRS changes its move sequence from $M_{Up}$ to $M_{Down}$.
     
          \item Move sequence  $M_{TurnD}$(Figure~\ref{4TurnOnTheDownWall}).
          By this move sequence, 
          the MRS changes its move sequence from $M_{Down}$ to $M_{Up}$.
     
          \item Move sequence $M_{B1}$(Figure~\ref{4MoveOnTheBottomOfTheWall}).
          By this move sequence, 
          the MRS changes its move sequence from $M_{Down}$ to $M_{B2}$.
     
          \item Move sequence $M_{B2}$(Figure~\ref{4MoveAlongTheDownWall}).
          The blue module is in cell $(x,y,0)$ at the start.
          By this move sequence, 
          the green module reaches cell $(x+1,y,0)$.
          By repeating $M_{B2}$ $n$ times, any of the modules visit the cells $(x+k,y,0)(0 \leq k \leq n)$. 
          That is, it visits all the cells of the horizontal line $\{(x,y,z)|y=s,z=0\}$. 
     
          \item Move sequence $M_{B3}$(Figure~\ref{4MoveOnTheBottomOfTheWallToRise}).
          By this move sequence, 
          the MRS changes its move sequence from $M_{B2}$ to $M_{Up}$.
     
          \item Move sequence $M_{Corner}$(Figure~\ref{4MoveOnTheBottomOfTheCorner}).
          By this move sequence, 
          the MRS changes its move sequence from $M_{Down}$ to $M_{B1}$.
     \end{itemize}

The proposed algorithm consists of the following six steps. 
We use north, south, east, and west for explanation, 
however each module does not need to know the directions. 

\begin{description}
     \item[Step 1] The MRS repeats $M_{Down}$,
           that makes it move in the down direction
          along a vertical line on a plane $\{(x,y,z) | y = s\}$ for some $s$.
     \item[Step 2] When the MRS reaches the bottom wall, 
            it changes the direction to up 
           by  $M_{TurnD}$.
     \item[Step 3] The MRS repeats $M_{Up}$,
           that makes it move in the up direction.
          This movement makes the MRS move along the same vertical line as Step $1$.
     \item[step 4] If the MRS is adjacent to the bottom of the west wall, 
     it moves to the next plane $\{(x,y,z) | y = s-1\}$
     by $M_{B1}$.
     Then it moves along the bottom wall in the east direction
     by repeating $M_{B2}$.
     When the MRS reaches the east wall,
     it performs $M_{B3}$.
     Then it starts searching the next plane by Step $1$.
     Otherwise, it proceeds to Step $5$. 
     \item [Step 5] When the MRS reaches the top wall, 
     it moves west by one row 
     by $M_{TurnUp}$.
     Then it returns to Step $1$.
     \item[Step 6]  When the MRS reaches the southwest corner, 
     it performs $M_{Corner}$,
     and performs Step $4$, 
     that makes it start searching the plane obtained by rotation the current 
     search plane by $90$ degrees. 
     Thus, the search planes are first perpendicular to the $x$ axis and
     the MRS moves to east. 
     Then, the search planes are perpendicular to the $y$ axis and the MRS moves to north. 
     Third the search planes are perpendicular to the $x$ axis and the MRS moves to west. 
     Finally, the search planes are perpendicular to the $y$ axis and the MRS moves to south.
        \end{description}
   
        Depending on its initial state,
        MRS may start from the middle of the above track.
         
        When the MRS is on a plane $\{(x, y, z) | y=s\}$ for some $s$, 
        it visits all cells on a vertical line $\{(x,y,z)|x=t,y=s\}$ by Steps $1$, $2$, and $3$.
        Then, it proceeds to a vertical line $\{(x,y,z)|x=t-1,y=s\}$ by Step $5$.
        By repeating Steps $1$, $2$, $3$, and $5$, it visits all cells on the plane 
        $\{(x,y,z) | y=s\}$.
        By Step $4$, it moves to the bottom of east wall in the next plane $y=s-1$,
        then it starts searching the plane.
        By repeating Steps $1$ to $5$, it eventually reaches the southwest corner of the bottom wall.
        It starts searching the vertical line \{$\{(x,y,z)|x=1,y=d-2\}$\} by Step $6$.
   
 \begin{figure}[htbp]
     \centering
     \includegraphics[keepaspectratio, width=\linewidth]
          {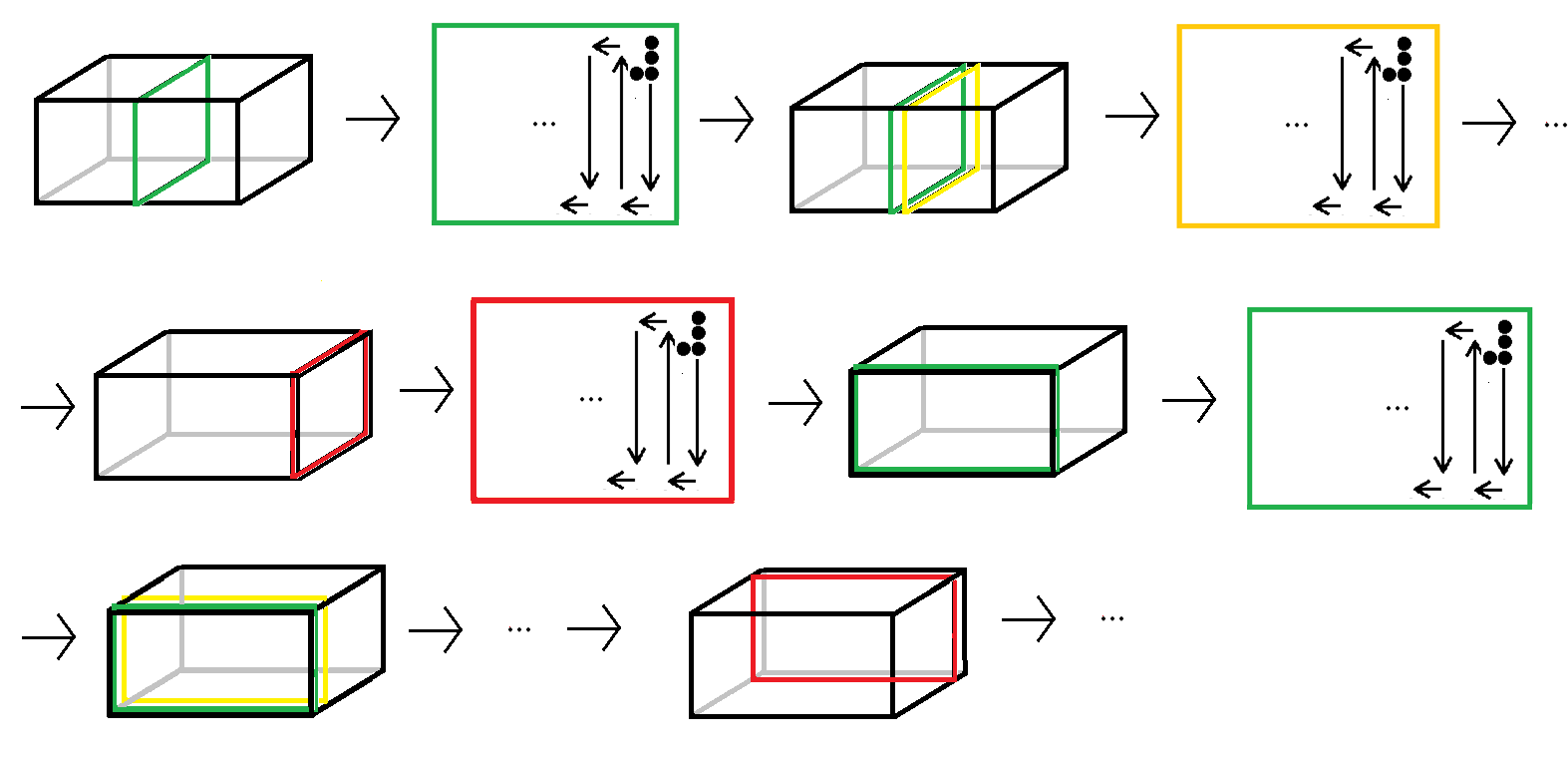}
     \caption{Example of a search by four modules}
     \label{Explode-Move4ap}
    \end{figure}
    
    \begin{figure}[htbp]
     \centering
     \includegraphics[keepaspectratio, height=2.5cm]
          {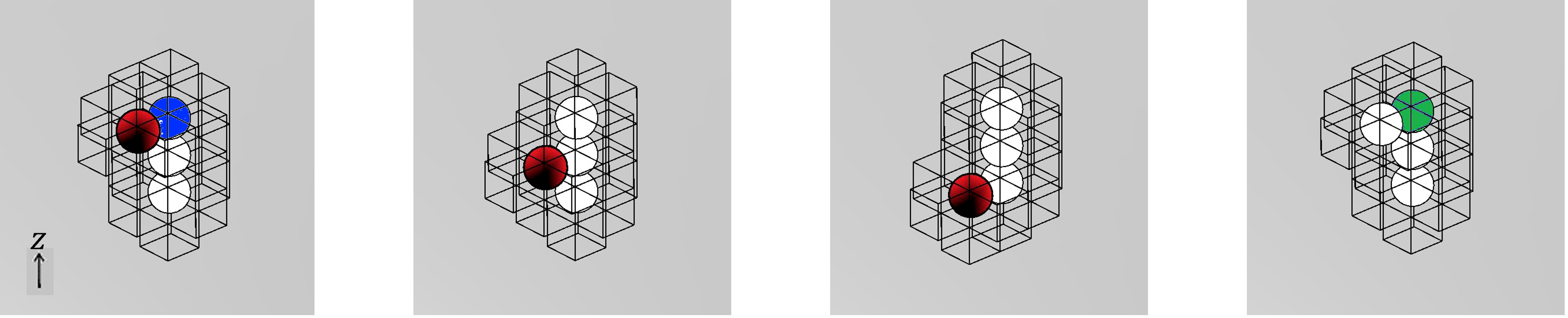}
     \caption{Move to down.
     In each figure, the red module moves. When the blue module is in cell $(x,y,z)$ at the start, after this move sequence, the green model reaches
     $(x,y,z-1)$.}
     \label{4MoveToDown}
    \end{figure}

   \begin{figure}[htbp]
     \centering
     \includegraphics[keepaspectratio, height=2.0cm]
          {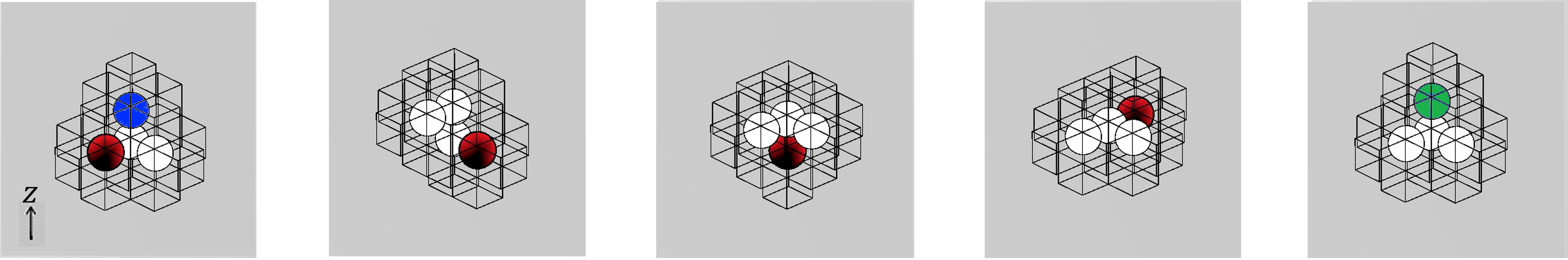}
     \caption{Move to up.
     In each figure, the red module moves. When the blue module is in cell $(x,y,z)$ at the start, after this move sequence, the green model reaches
     $(x,y,z+1)$. }
     \label{4MoveToUp}
    \end{figure}

    \begin{figure}[htbp]
     \centering
     \includegraphics[keepaspectratio, height=2.5cm]
          {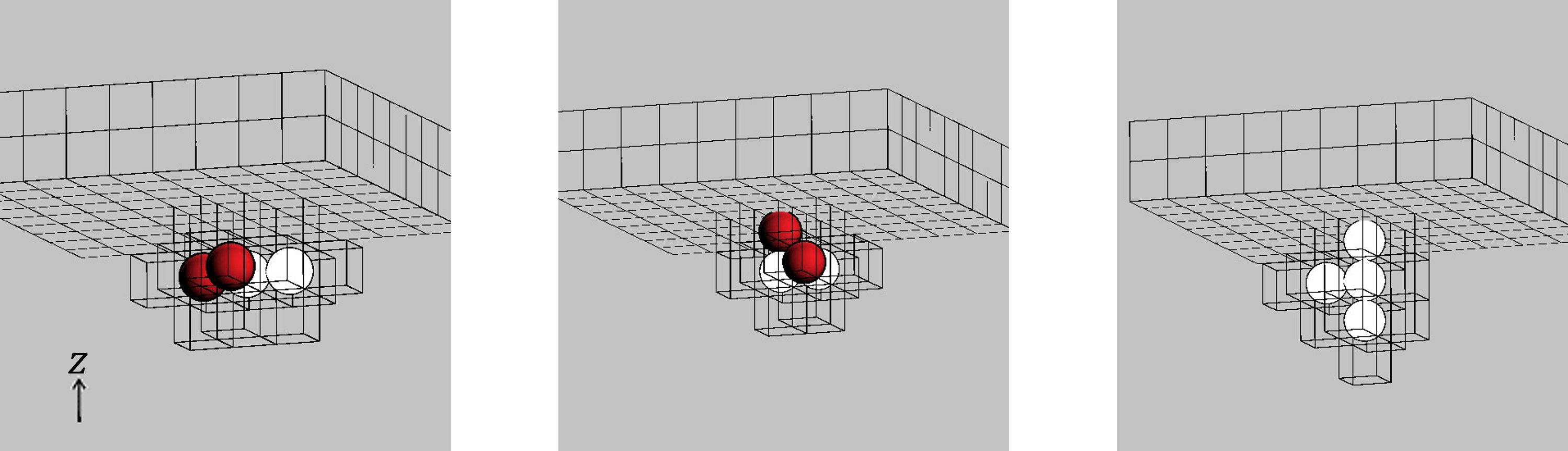}
     \caption{Turn on the up wall. In each figure, the red modules move.}
     \label{4TurnOnTheUpWall}
    \end{figure}

    \begin{figure}[htbp]
     \centering
     \includegraphics[keepaspectratio, height=2.5cm]
          {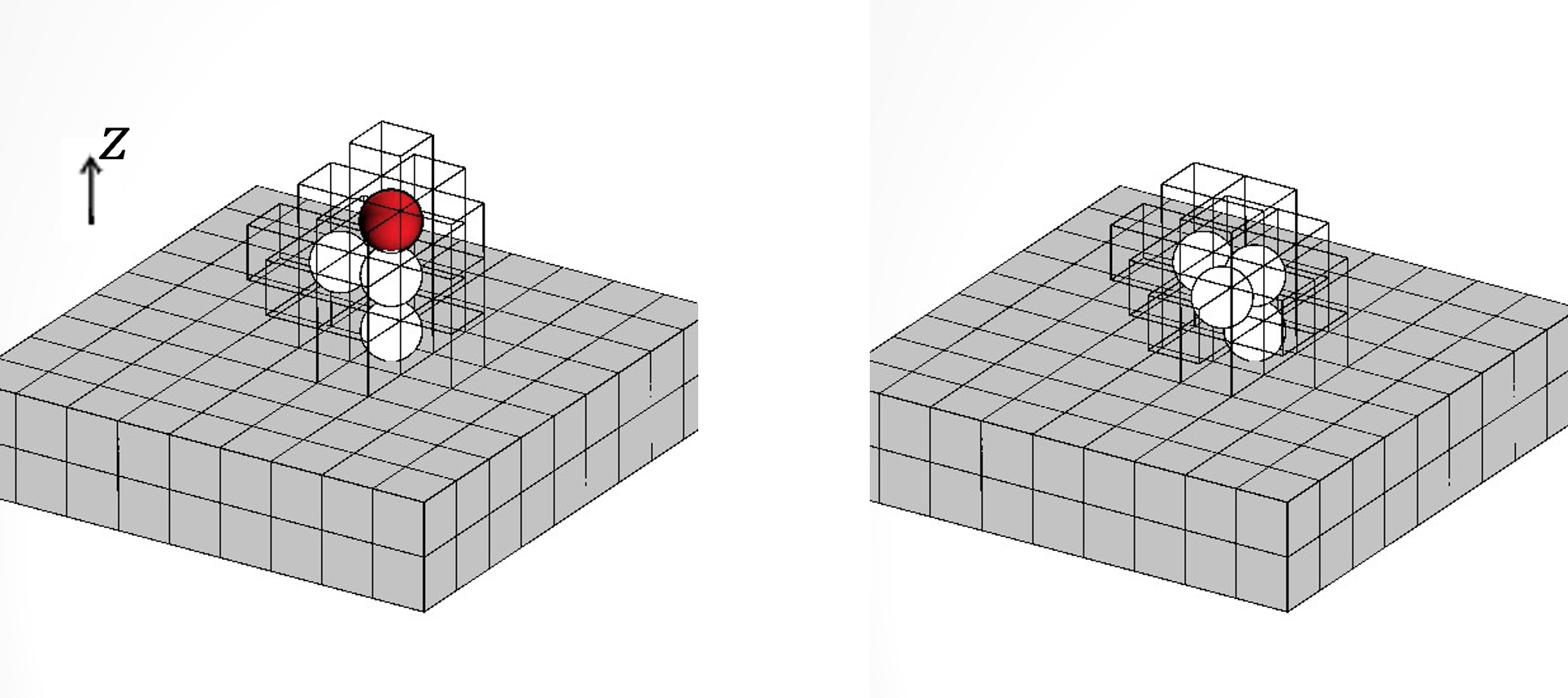}
     \caption{Turn on the down wall. In the first figure, the red module moves.}
     \label{4TurnOnTheDownWall}
    \end{figure}

    \begin{figure}[htbp]
     \centering
     \includegraphics[keepaspectratio, height=2.0cm]
          {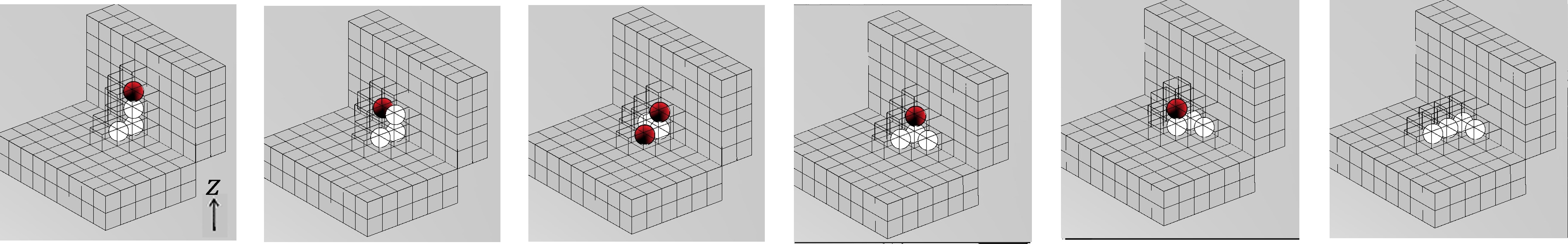}
     \caption{Move on the bottom of the wall. In each figure, the red modules move.}
     \label{4MoveOnTheBottomOfTheWall}
    \end{figure}

    \begin{figure}[htbp]
     \centering
     \includegraphics[keepaspectratio, height=2.5cm]
          {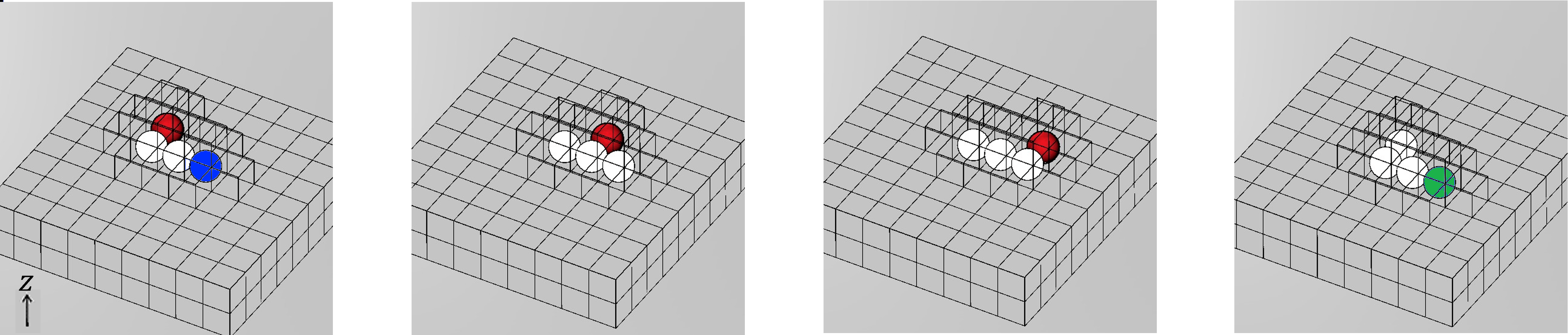}
     \caption{Move along the down wall.
      In each figure, the red module moves. When the blue module is in cell $(x,y,0)$ at the start, after this move sequence, the green model reaches
     $(x+1,y,0)$.}
     \label{4MoveAlongTheDownWall}
    \end{figure}

    \begin{figure}[htbp]
     \centering
     \includegraphics[keepaspectratio, height=2.5cm]
          {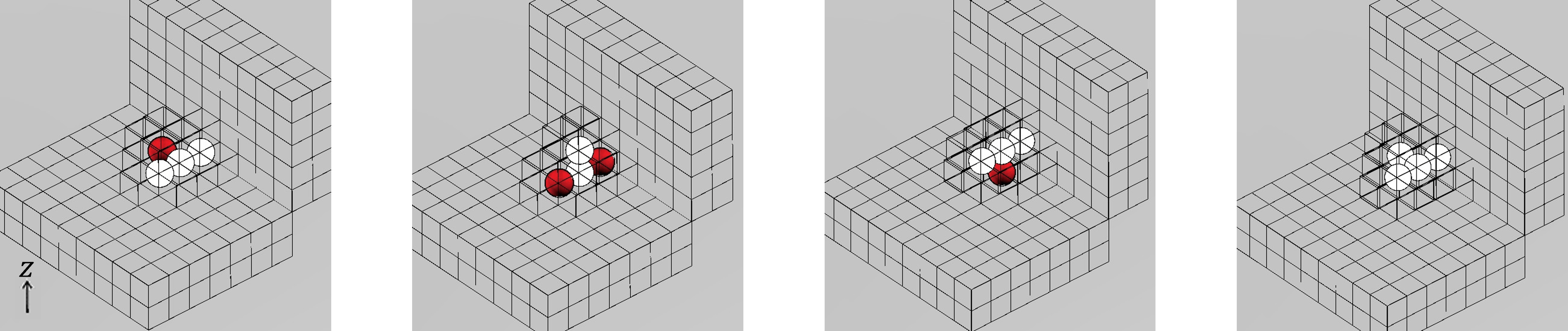}
     \caption{Move on the bottom of the wall to rise. In each figure, the red modules move.}
     \label{4MoveOnTheBottomOfTheWallToRise}
    \end{figure}

    \begin{figure}[htbp]
     \centering
     \includegraphics[keepaspectratio, height=2.5cm]
          {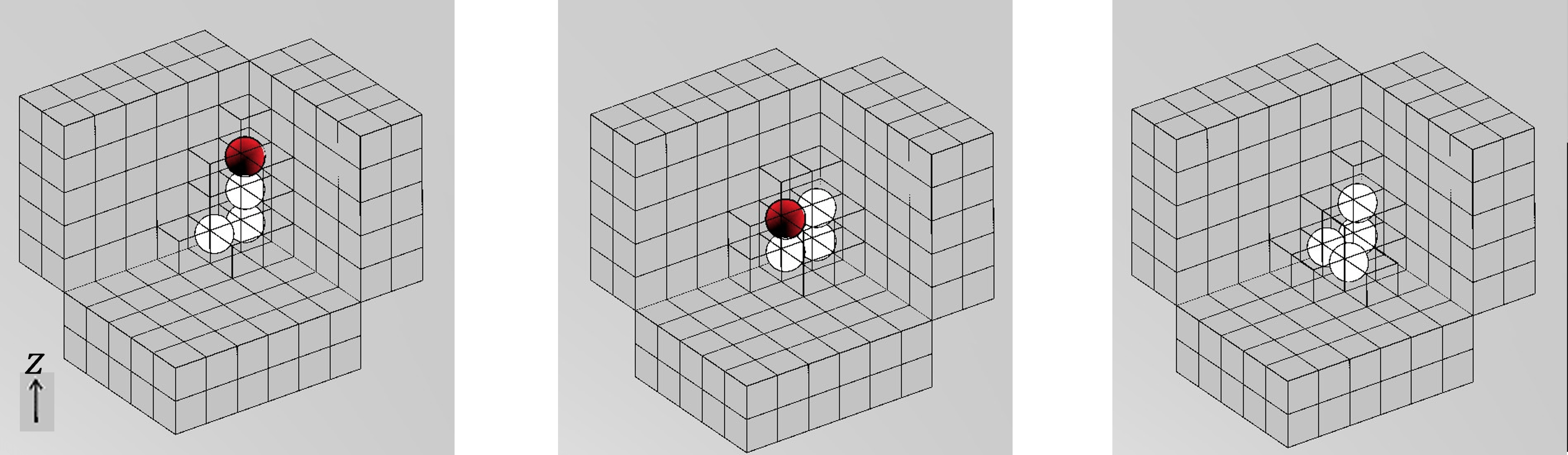}
     \caption{Move on the bottom corner. In each figure, the red module moves.}
     \label{4MoveOnTheBottomOfTheCorner}
    \end{figure}
    
  \begin{figure}[htbp]
     \centering
     \includegraphics[keepaspectratio, scale=0.5]
          {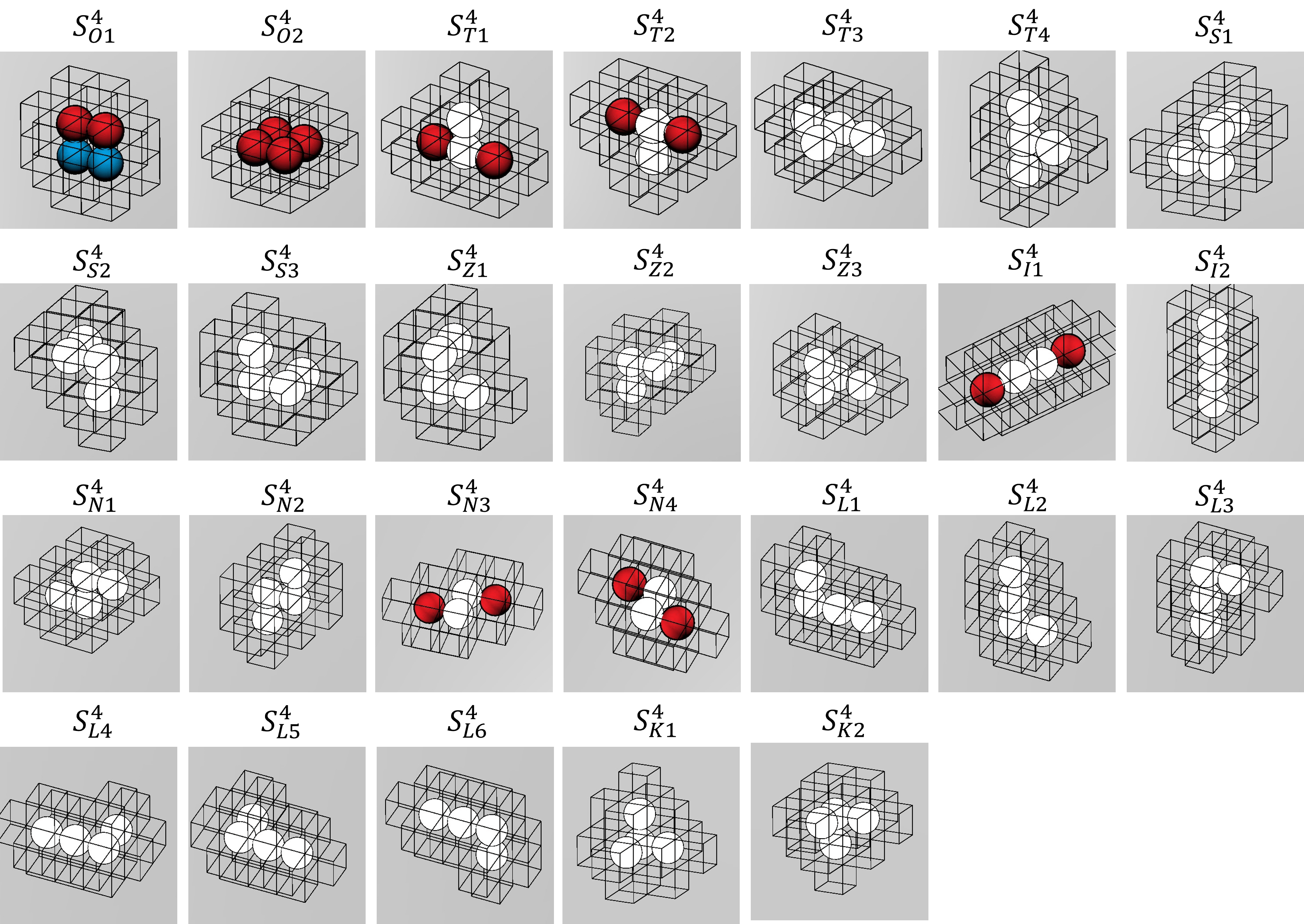}
     \caption{Configurations of the MRS of four modules equipped with a common vertical axis}
     \label{4ModuleShapes}
    \end{figure}

    \begin{figure}[htbp]
     \centering
     \includegraphics[keepaspectratio, scale=0.5]
          {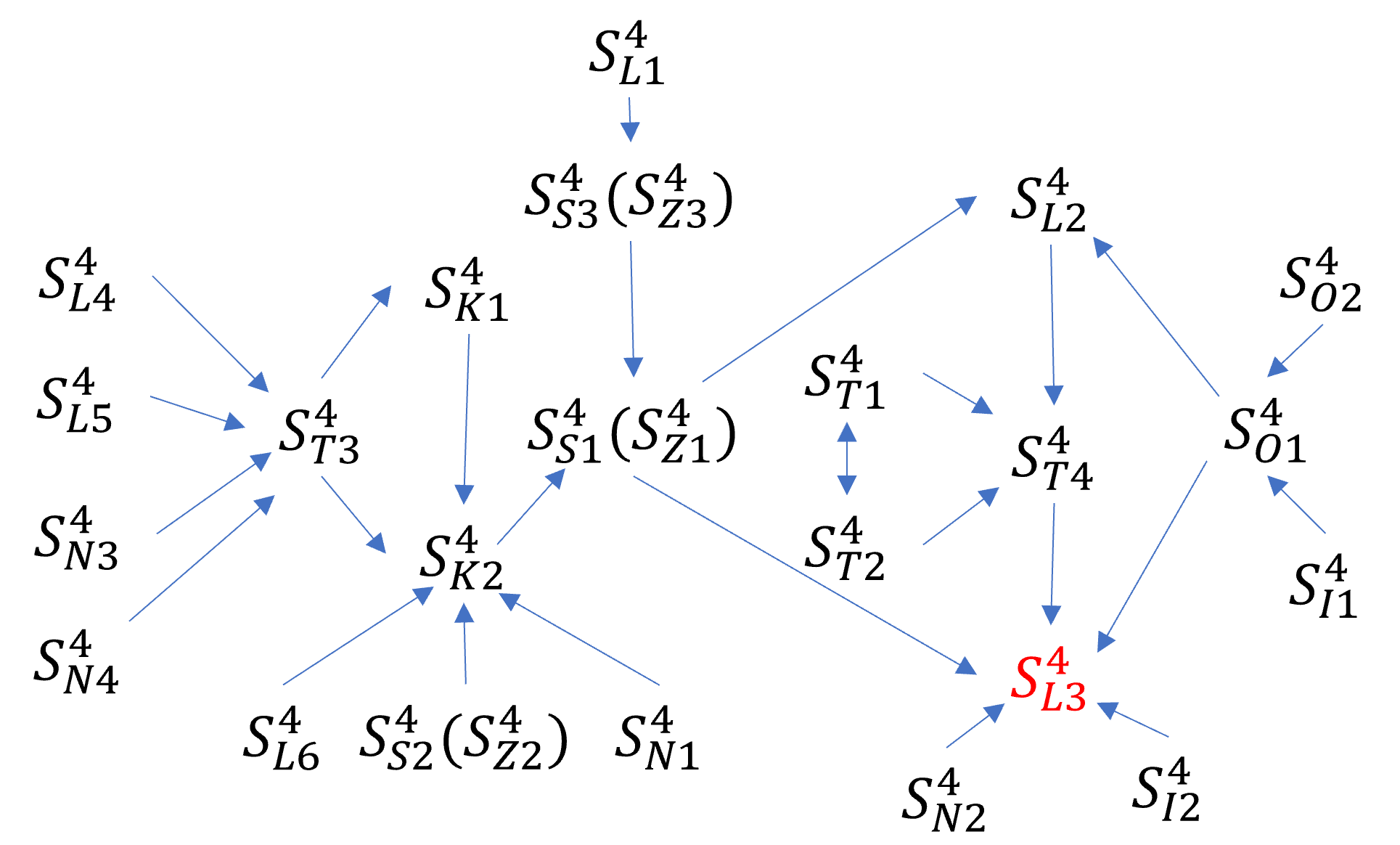}
     \caption{Transition graph of the MRS of four modules equipped with a common vertical axis}
     \label{4ModuleTransitionGraph}
    \end{figure}
    \clearpage
    \begin{figure}[htbp]
       \centering
       \includegraphics[keepaspectratio, height=2cm]
            {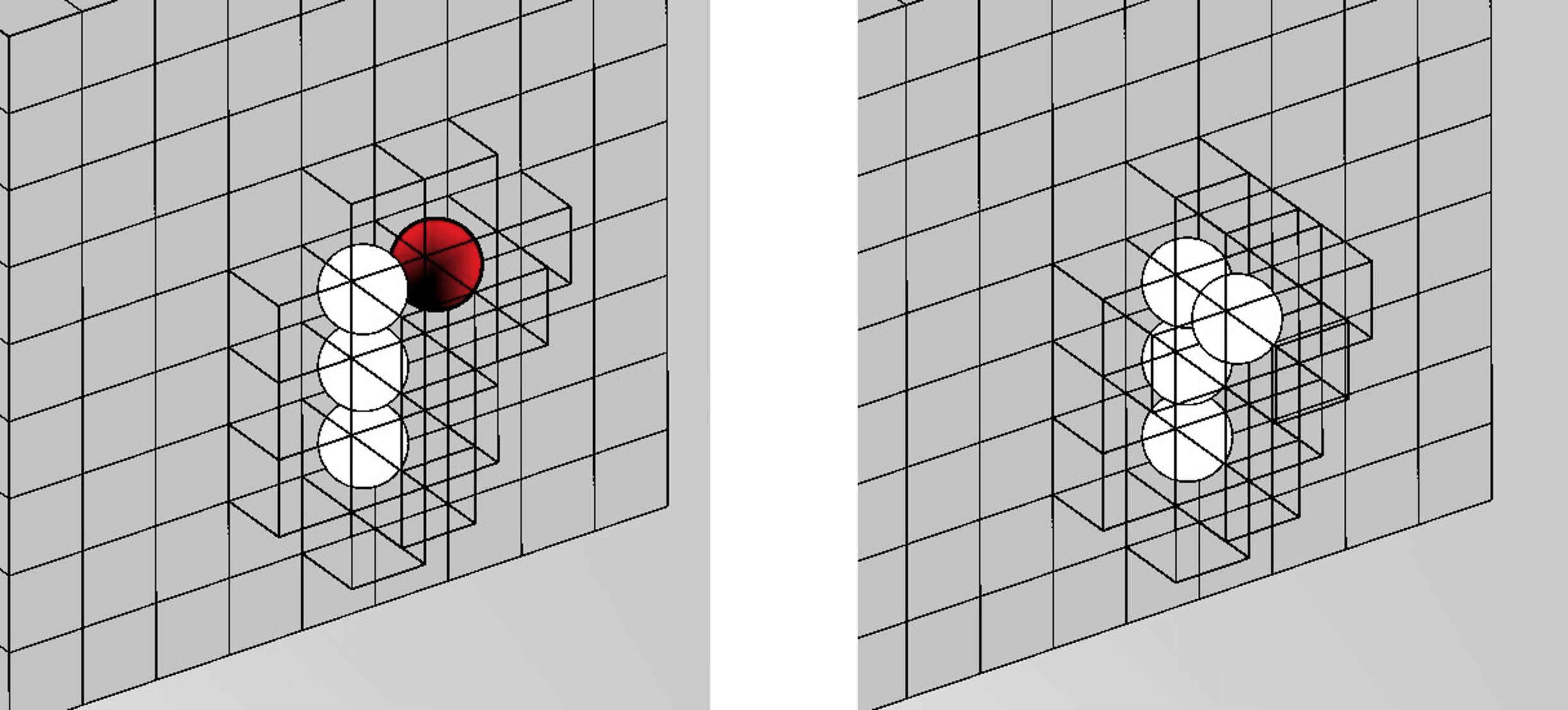}
       \caption{Transformation of $S^4_{L3}$}
       \label{Ltransform}
      \end{figure}

      Depending on the initial configuration, the search may not be possible, 
and even if it is possible, the MRS cannot start the search directly because some initial configurations do not satisfy the 
conditions in Table~\ref{tab_4module}. 

Figure~\ref{4ModuleShapes} shows all possible configurations for four modules not equipped with a common vertical axis.
The configurations containing the coloured modules are symmetric, i.e., 
In each figure the modules painted by the same color may 
have an identical observation and move symmetrically. 
The MRS cannot start the search because it cannot move from their positions.
Therefore, if there are pairs of modules with the same observation in an initial configuration, 
the search is not possible.

The proposed algorithm guarantees that the MRS finds the target 
when all modules have different observations in an initial configuration. 
Figure~\ref{4ModuleTransitionGraph} shows these initial configurations can be transformed to $S^4_{L3}$ by a sequence of transformations. 
Even if some walls prevents these transformations, 
the MRS can leave the wall because its state is asymmetric. 
If the MRS satisfies the condition of the fourth movement of Table~\ref{tab_4module}, 
the MRS can start the search directly.
Otherwise, the MRS cannot move because it touches a wall. 
In this case, as shown in Figure~\ref{Ltransform}, the MRS in $S^4_{L3}$ on a wall 
can change its direction by a single rotation
and in this new configuration it can perform one of these movements. 
By adding the above rules to the algorithm, the MRS can start the search.

 \begin{longtable}{|l|c|c|c|c|}
     \caption{Search algorithm for the MRS of four modules equipped with a common vertical axis}
     \endhead
     \endfoot
   \cline{1-5}          & $C_m$& $C_w$ & $C_e$ & Output \\
   \cline{1-5}    $M_{Down}$ &  $(0,1,0)$,$(0,1,1)$,   &        &  $(0,0,-2)$ &          $(0,0,-1)$ \\
                               &  $(0,1,-1)$   &        &  &         \\
   \cline{2-5}          &  $(0,1,0)$,$(0,1,1)$,   &         &  $(0,0,-1)$ &         $(0,1,-1)$ \\
                       &  $(0,1,2)$   &         &  &         \\
   \cline{2-5}            &  $(0,0,-1)$,$(0,0,-2)$,    &     &  $(0,0,-3)$ &         $(0,-1,-1)$ \\
          &  $(0,-1,-2)$    &     &  &        \\
   \cline{2-5}          &  $(0,1,0)$,$(0,1,-1)$,   &         &   $(1,0,0)$ &          $(0,0,-1)$ \\
         &  $(0,1,-2)$   &         &   &           \\
   \cline{1-5}    $M_{TurnD}$ &  $(0,0,-1)$,$(-1,0,-1)$, &  $(0,0,-3)$           &   &       $(0,-1,-1)$ \\
   &  $(0,0,-2)$ &  $(0,0,-3)$           &   &      \\
   \cline{1-5}     $M_{Up}$ &  $(-1,0,1)$,$(0,-1,1)$,    &   &   &          $(1,0,1)$ \\
    &  $(0,0,1)$    &   &   &        \\
    \cline{2-5}         &  $(-1,0,0)$,$(-2,0,0)$,   &     &  $(0,0,2)$,&        $(-1,0,1)$ \\
                            &             $(-1,-1,0)$                      &            & $(0,0,-1)$ &      \\
    \cline{2-5}         &  $(1,0,0)$,$(1,-1,0)$,    &         &  $(0,0,-1)$ &          $(0,0,1)$ \\
           &  $(1,0,1)$    &         &   &        \\
    \cline{2-5}         &  $(0,1,0)$,$(0,1,1)$,    &        &  $(0,0,-1)$ &          $(0,0,1)$ \\
        &  $(-1,1,1)$    &        &  &          \\
    \cline{1-5} $M_{TurnU}$  &  $(0,1,0)$,$(1,1,0)$, &  $(0,0,2)$   &       &          $(1,0,0)$ \\
                               &  $(-1,1,0)$ &   &       &           \\
    \cline{2-5}        &  $(1,0,0)$,$(2,0,0)$, &  $(0,0,2)$    &         &         $(1,0,1)$ \\
            &  $(1,-1,0)$ &  $(0,0,2)$    &         &        \\
    \cline{2-5}       &  $(0,0,-1)$,$(1,0,-1)$, &  $(0,0,1)$  &         &          $(1,0,0)$ \\
        &  $(1,-1,-1)$ &   &         &        \\
    \cline{2-5}        &  $(0,1,0)$,$(-1,1,0)$, &  $(0,0,2)$  &       &            $(0,1,-1)$ \\
            &  $(-1,1,1)$ &    &       &            \\
    \cline{1-5}     $M_{B1}$  &  $(0,1,0)$,$(0,1,1)$, &  $(0,2,0)$,$(0,0,-2)$        &   &          $(0,0,-1)$ \\
     &  $(0,1,-1)$ &         &   &           \\
    \cline{2-5}   &  $(0,0,-1)$,$(0,0,-2)$, &  $(0,1,0)$,$(0,0,-3)$        &  $(-1,0,0)$ &          $(-1,0,-1)$ \\
     &  $(0,-1,-2)$ &        &   &          \\
    \cline{2-5}   &  $(1,0,0)$,$(1,0,-1)$, &  $(0,1,0)$,$(0,0,-2)$        &   &        $(0,0,-1)$ \\
       & $(1,-1,-1)$ &         &   &         \\
    \cline{2-5}       &  $(0,0,-1)$,$(-1,0,-1)$, &  $(0,1,0)$,$(0,0,-2)$       &  $(1,0,0)$ &          $(-1,0,0)$ \\
      & $(0,-1,-1)$ &        &   &          \\
    \cline{2-5}        &  $(0,1,0)$,$(-1,1,0)$, &  $(0,2,0)$,$(0,0,-1)$       &  $(1,0,0)$ &          $(-1,0,0)$ \\
       & $(0,1,1)$ &      &   &           \\
    \cline{2-5}        &  $(0,0,-1)$,$(1,0,-1)$, &  $(0,1,0)$   &        &        $(0,-1,0)$ \\
           &  $(0,-1,-1)$ &   &        &      \\
    \cline{2-5}        &  $(0,0,-1)$,$(0,1,-1)$, &  $(0,1,0)$,$(0,0,-2)$      &   &       $(0,-1,-1)$ \\
         &  $(1,1,-1)$ &     &   &       \\
    \cline{1-5}    $M_{B2}$ &  $(0,1,0)$,$(-1,1,0)$, &  $(0,0,-1)$   &          &         $(-1,0,0)$ \\
     &  $(-2,1,0)$ &    &          &        \\
    \cline{2-5}      &  $(0,1,0)$,$(1,1,0)$, &  $(0,0,-1)$   &       $(-2,0,0)$ &       $(-1,0,0)$ \\
         &  $(-1,1,0)$ &   &      &        \\
    \cline{2-5}         &  $(0,1,0)$,$(1,1,0)$, &  $(0,0,-1)$   &          &        $(-1,1,0)$ \\
            &  $(2,1,0)$ &     &          &       \\
    \cline{2-5}         &  $(-1,0,0)$,$(-2,0,0)$, &  $(0,0,-1)$   &          &         $(-1,-1,0)$ \\
           &  $(-2,-1,0)$ &   &          &      \\
   \cline{1-5}    $M_{B3}$  &  $(0,1,0)$,$(1,1,0)$,&  $(0,0,-1)$,$(-2,0,0)$         &    &      $(0,1,1)$ \\
    &  $(-1,1,0)$ &          &    &      \\
   \cline{2-5}  &  $(1,0,0)$,$(2,0,0)$, &  $(0,0,-1)$,$(-1,0,0)$       &   &          $(0,0,1)$ \\
     &  $(1,0,1)$ &         &   &         \\
   \cline{2-5}        &  $(-1,0,0)$,$(-2,0,0)$, &  $(0,0,-1)$,$(-3,0,0)$      &   &         $(0,0,1)$ \\
        &  $(-1,0,1)$ &     &   &       \\
   \cline{2-5}        &  $(0,0,1)$,$(1,0,1)$, &  $(0,0,-1)$,$(-2,0,0)$       &   &         $(0,-1,1)$ \\
         &  $(-1,0,1)$ &      &   &       \\
   \cline{1-5}    $M_{Corner}$ &  $(0,0,-1)$,$(0,0,-2)$, &  $(0,1,0)$,$(1,0,0)$,&   &       $(-1,0,-1)$ \\
                                            & $(-1,0,-2)$    &$(0,0,-3)$      &      &        \\
     \cline{2-5}       &  $(1,0,0)$,$(0,0,-1)$, &  $(0,1,0)$,$(2,0,0)$,&   &        $(0,-1,-1)$ \\
                          &     $(1,0,-1)$                              &  $(0,0,-2)$ &     &        \\
     \cline{2-5}        &  $(0,0,-1)$,$(-1,0,-1)$,&  $(1,0,0)$,$(0,1,0)$,& $(0,0,-2)$    &         $(-1,0,0)$ \\
                         &                         $(-1,-1,-1)$                   &  $(0,0,-2)$        &     &         \\
   \cline{1-5}          & \multicolumn{1}{c|}{Otherwise} & \multicolumn{1}{c|}{Otherwise} & \multicolumn{1}{c|}{Otherwise} & $(0,0,0)$ \\
   \cline{1-5}   
   \end{longtable}
     \label{tab_4module}%

\subsection{Search without a common compass} 

We show the following theorem by a search algorithm for the MRS of 
five modules not equipped with a common compass. 

\begin{theorem}
\label{theorem: 5 modules can Explore 3D space without compass}
The MRS of five modules not equipped with a common compass can solve a search problem 
in a finite 3D grid 
if no pair of modules have an identical observation in an initial configuration. 
\end{theorem}

We prove Theorem~\ref{theorem: 5 modules can Explore 3D space without compass} 
by a search algorithm for the MRS of three modules not equipped with a common compass.  

The proposed algorithm considers each plane 
perpendicular to one of the $x$, $y$, and $z$ axis. 
The choice of the axis depends on the initial configuration of the MRS, and 
the modules do not need to know the global coordinate system. 
In the following, without loss of generality, we assume that the MRS considers 
planes perpendicular to the $x$ axis, i.e., 
$x=s$($y=s$,$z=s$, respectively) for $s = 0, 1, 2, \ldots$. 
It moves along each vertical line $\{(x,y,z) | x=s, y=t, z=u\}$ or 
horizontal line $\{(x,y,z) | x=s, y=u,z=t\}$ for $u= 0, 1, 2, \ldots$ 
on the plane.
Figure~\ref{Explode-Move5} shows an execution of the algorithm. 

The MRS continues to search each plane perpendicular to the $x$-axis 
until it reaches the east wall. 
Then, it changes the search direction from $x^+$ direction to 
$x^-$ direction and 
it starts to search each plane perpendicular to the $x$-axis 
until it reaches the west wall. 

The algorithm description is given in Table~\ref{tab_5modules}.

    \begin{figure}[htbp]
      \centering
      \includegraphics[keepaspectratio,width=\linewidth]
           {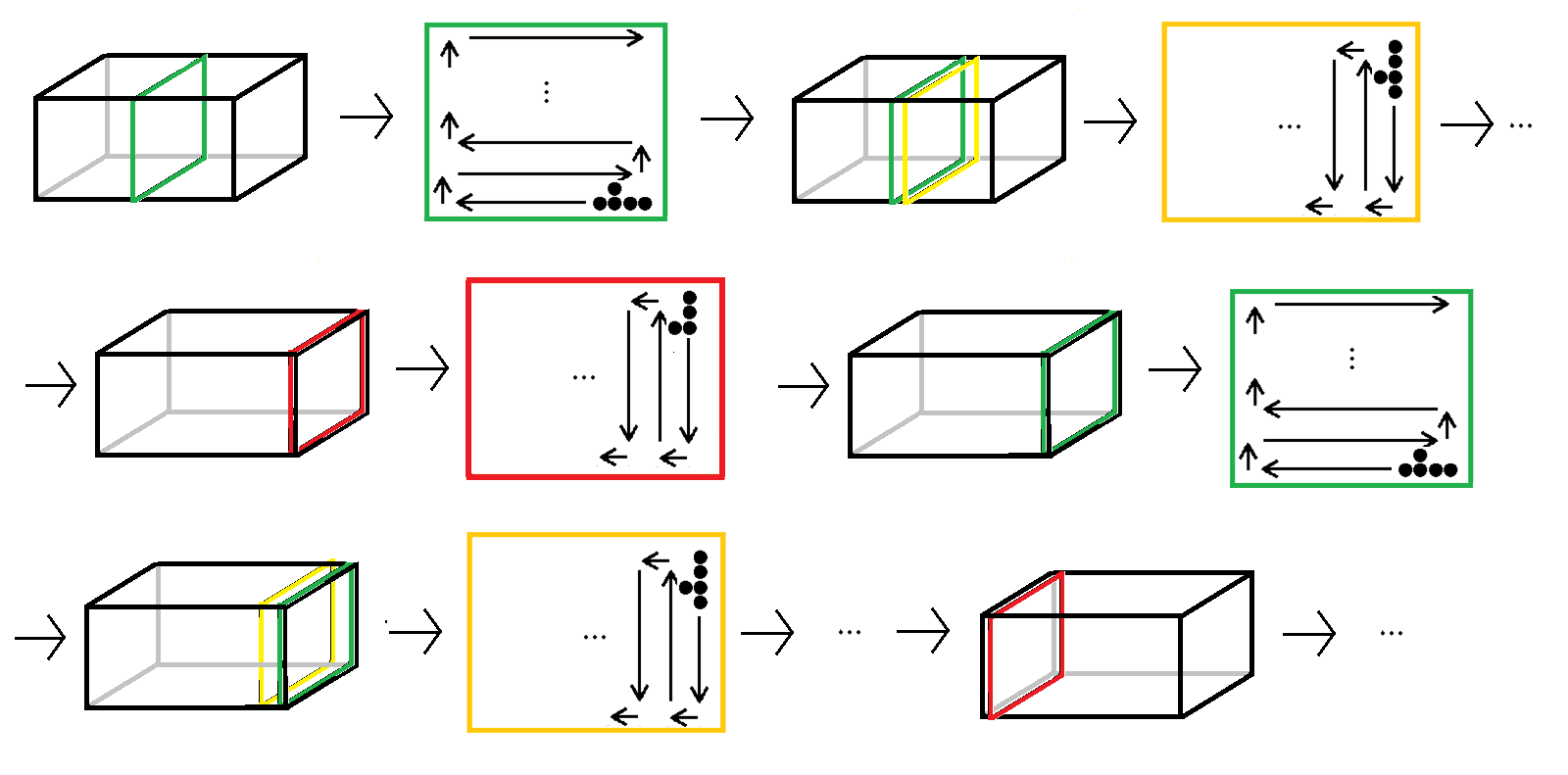}
      \caption{Example of a search with five modular robots}
      \label{Explode-Move5}
     \end{figure}

The proposed algorithm consists of the following move sequences.
\begin{itemize}
     \item Move sequence $M_{Forward}$(Figure~\ref{5MoveToForward}).
     The blue module is in cell $(x,y,z)$ at the start.
     By this move sequence, 
     the green module reaches cell $(x,y,z-1)$.
     By repeating $M_{Forward}$ $n$ times, any of the modules visit the cells $(x,y,z-k)(0 \leq k \leq n)$. 
     That is, it visits all the cells of the horizontal line  $\{(x,y,z)|x=s,y=t\}$.
     
     \item Move sequence $M_{Back}$(Figure~\ref{5MoveToBack}).
     The blue module is in cell $(x,y,z)$ at the start.
     By this move sequence, 
     the green module reaches cell $(x,y,z+1)$.
     By repeating $M_{Back}$ $n$ times, any of the modules visit the cells $(x,y,z+k)(0 \leq k \leq n)$. 
     That is, it visits all the cells of the horizontal line  $\{(x,y,z)|x=s,y=t\}$.
     
     \item Move sequence $M_{TurnB}$(Figure~\ref{5TurnToBack}).
     By this move sequence, 
     the MRS changes its move sequence from $M_{Forward}$ to $M_{Back}$.
     
     \item Move sequence $M_{TurnF}$(Figure~\ref{5TurnToForward}).
     By this move sequence, 
     the MRS changes its move sequence from $M_{Back}$ to $M_{Forward}$.
     
     \item Move sequence $M_{Edge}$(Figure~\ref{5MoveOnTheEdge}).     
     By this move sequence, 
     the MRS changes its move sequence from $M_{Forward}$ to $M_{TurnB}$.
     
     \item Move sequence $M_{Corner}$(Figure~\ref{5MoveOnTheCorner}).
     By this move sequence, 
     the MRS changes its move sequence from $M_{Forward}$ to $M_{TurnB}$.
     \end{itemize}
   
The proposed algorithm consists of the following six steps. 
We use down direction for explanation, but each module does not need to know the down direction.

\begin{description}
     \item [Step 1]  The MRS repeats $M_{Forward}$,
     that makes it move to the down direction 
     along a vertical line on a plane $\{(x,y,z) | x = s\}$ for some $s$.
     \item [Step 2] When the MRS reaches the bottom wall, 
     it changes the direction to up by $M_{TurnB}$.
     \item[Step 3] The MRS repeats $M_{Back}$,
     that makes it move in the up direction along a vertical line followed in Step $1$. 
     \item[Step 4] If the MRS is adjacent to the bottom of the west wall, 
     it moves to plane $\{(x,y,z) | x = s+1\}$
     by $M_{TurnF}$,
     and starts searching the new plane by repeating Steps $1$, $2$, and $3$. 
     \item[step 5] When the MRS reaches the top wall or the bottom wall, 
     it moves to the southern row
     by $M_{Edge}$.
     Then it repeats Steps $1$, $2$, and $3$ again.
     \item[Step 6]  When the MRS reaches a corner of the field, 
      perform $M_{Corner}$
     that makes it changes the search direction from the positive $x$ direction to 
     the negative $x$ direction.  
     \end{description}
     When the MRS is on a plane $\{(x, y, z) | x=s\}$ for some $s$, 
     it visits all the cells on line $\{(x,y,z)|x=s,y=t$\} by Steps $1$, $2$, and $3$.
     Then, it proceeds to a new line \{$\{(x,y,z)|x=s,y=t-1,z=u\}$\} by Step $5$.
     By repeating Steps $1$, $2$, $3$, and $5$, it visits all cells on plane $x=s$.
     By Step $4$, it starts searching a new plane $x=s+1$.
     By repeating Step $1$ to $5$, it eventually reaches the corner adjacent to the south wall 
     and the east wall. 
     By Step $6$, it starts searching west wall. 
     By repeating Step $1$ to $6$ twice, the MRS visits all cells of the field. 
     
         \begin{figure}[htbp]
           \centering
           \includegraphics[keepaspectratio,width=\linewidth]
                {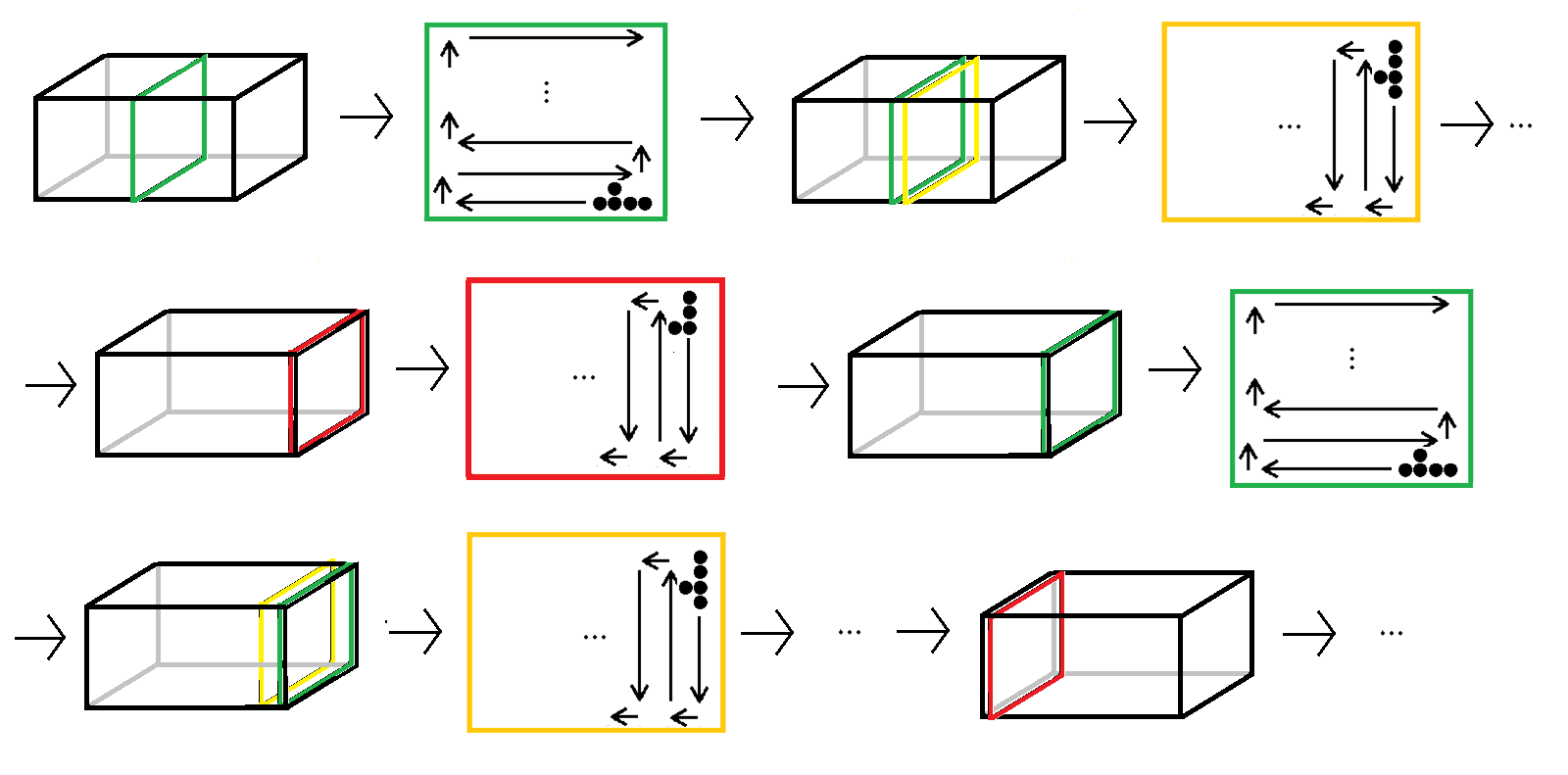}
           \caption{Example of a search with five modular robots}
           \label{Explode-Move5ap}
          \end{figure}
          
           \begin{figure}[htbp]
             \centering
             \includegraphics[keepaspectratio, height=2cm]
                  {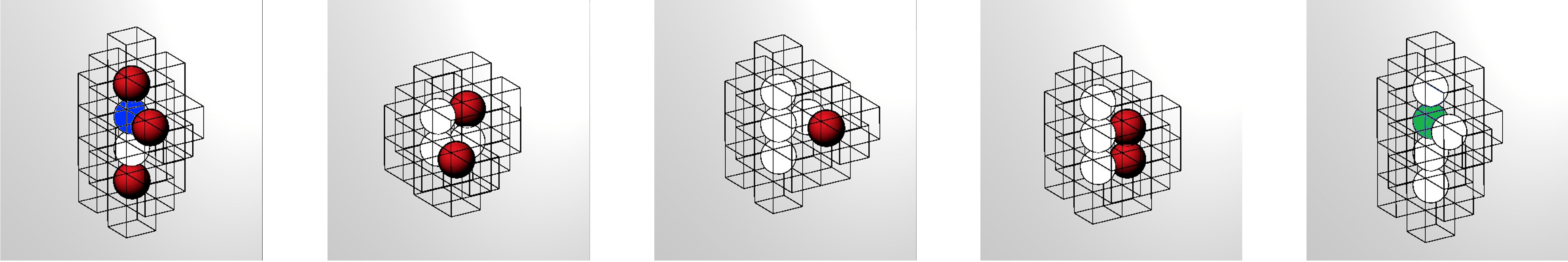}
             \caption{Move to forward.
             In each figure, the red modules move. When the blue module is in cell $(x,y,z)$ at the start, after this move sequence, the green model reaches
           $(x,y,z-1)$.}
             \label{5MoveToForward}
            \end{figure}
       
            \begin{figure}[htbp]
             \centering
             \includegraphics[keepaspectratio, height=2.5cm]
                  {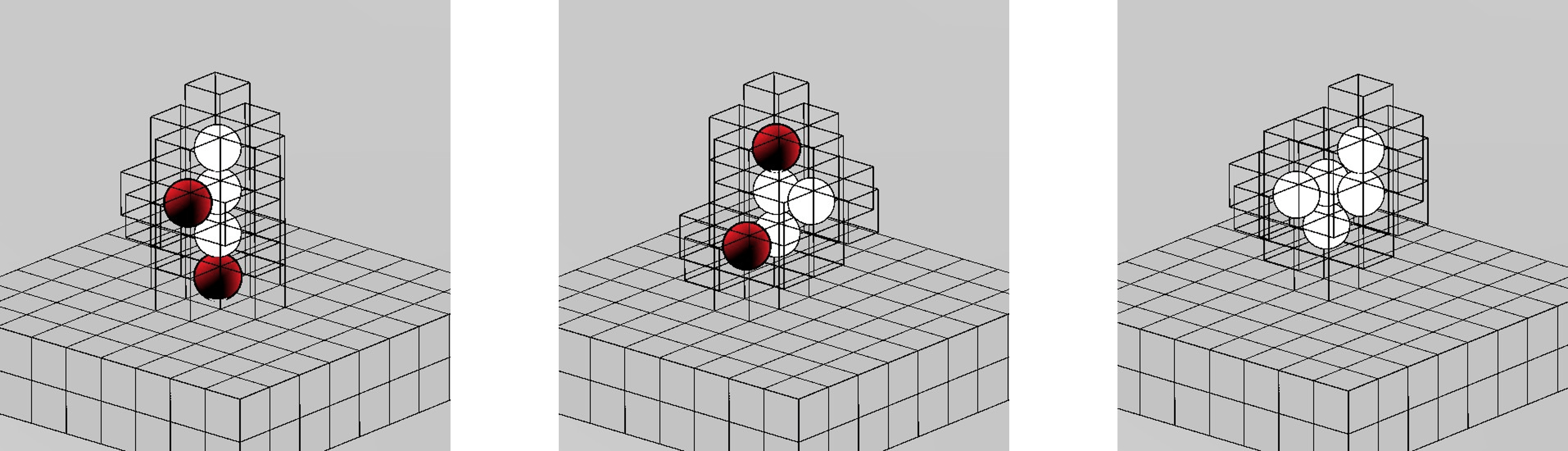}
             \caption{Turn to back. In each figure, the red modules move.}
             \label{5TurnToBack}
            \end{figure}
       
            \begin{figure}[htbp]
             \centering
             \includegraphics[keepaspectratio, height=2cm]
                  {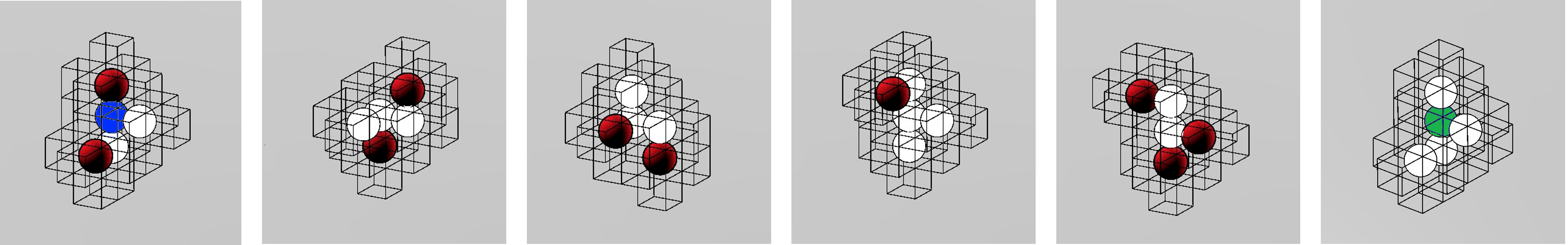}
             \caption{Move to back. In each figure, the red modules move. When the blue module is in cell $(x,y,z)$ at the start, after this move sequence, the green model reaches
              $(x,y,z+1)$.}
             \label{5MoveToBack}
            \end{figure}
       
            \begin{figure}[htbp]
             \centering
             \includegraphics[keepaspectratio, height=5cm]
                  {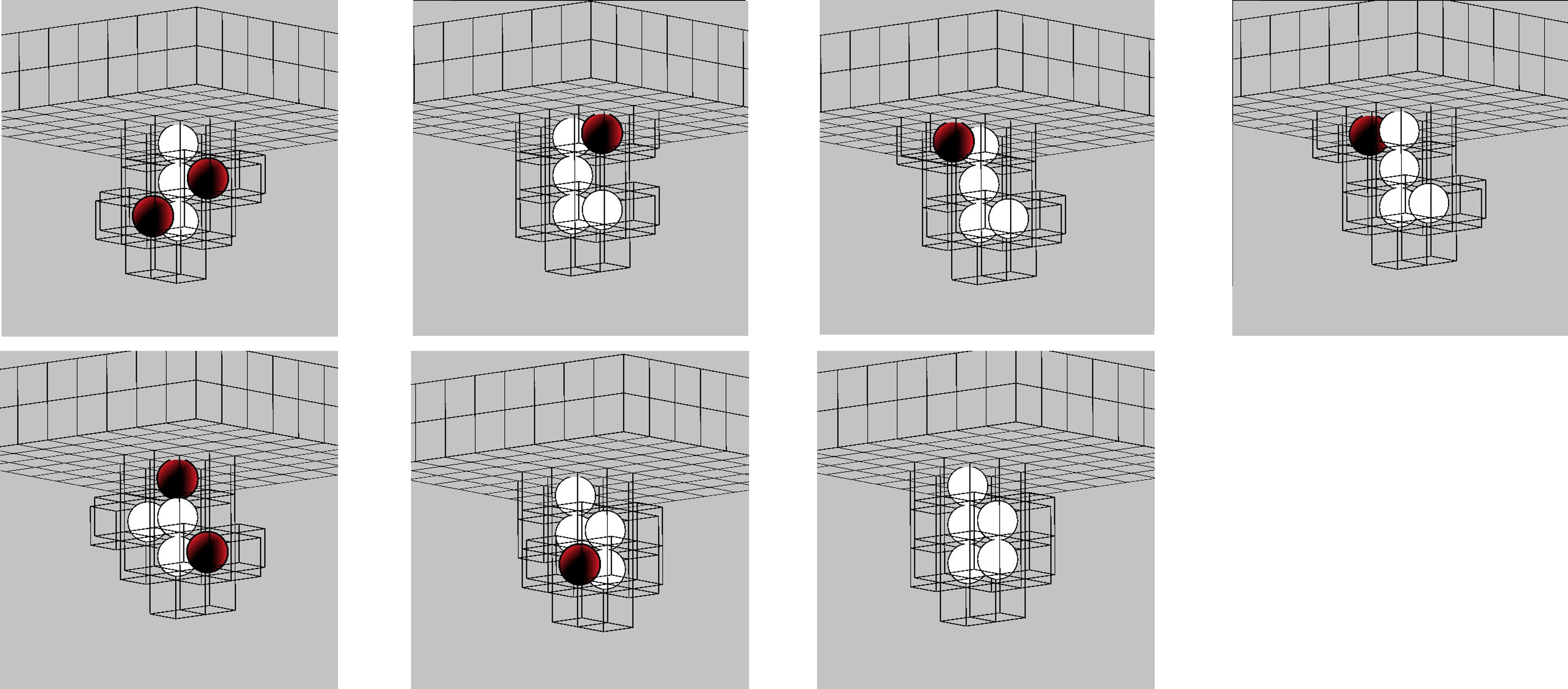}
             \caption{Turn to forward. In each figure, the red modules move.}
             \label{5TurnToForward}
            \end{figure}
       
            \begin{figure}[htbp]
             \centering
             \includegraphics[keepaspectratio, height=5cm]
                  {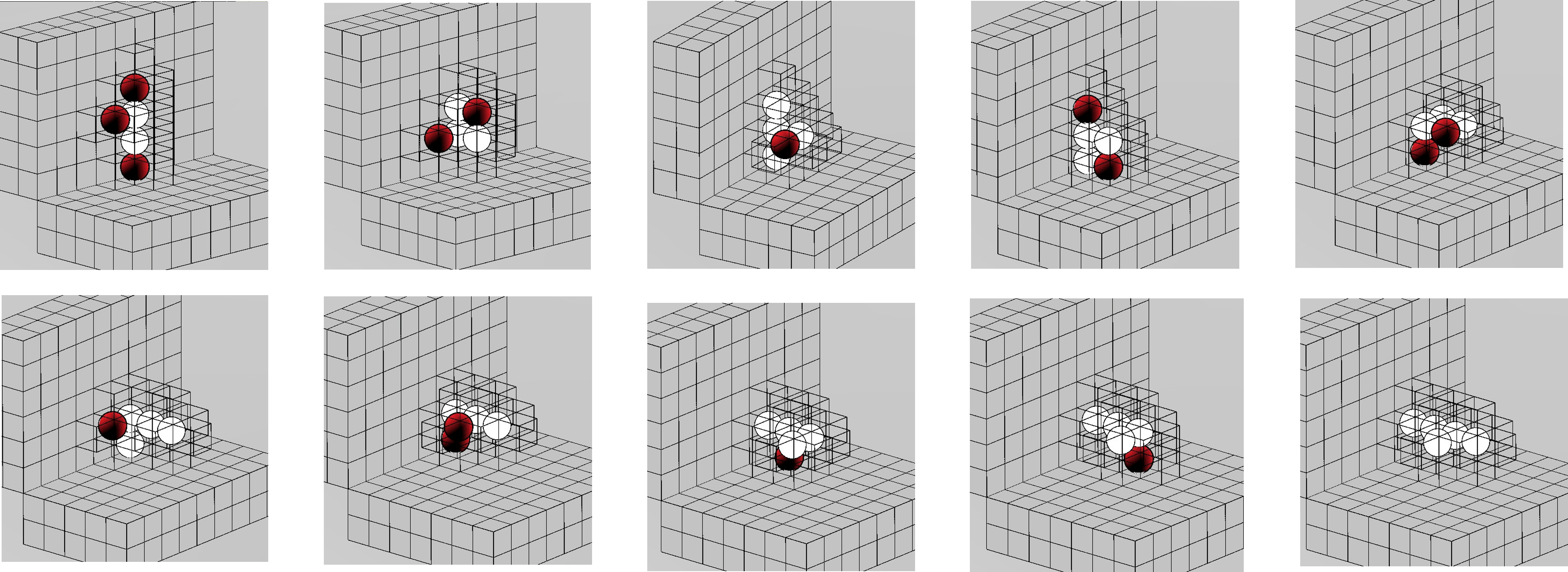}
             \caption{Move on the edge. In each figure, the red modules move.}
              \label{5MoveOnTheEdge}
            \end{figure}
       
            \begin{figure}[htbp]
             \centering
             \includegraphics[keepaspectratio, height=5cm]
                  {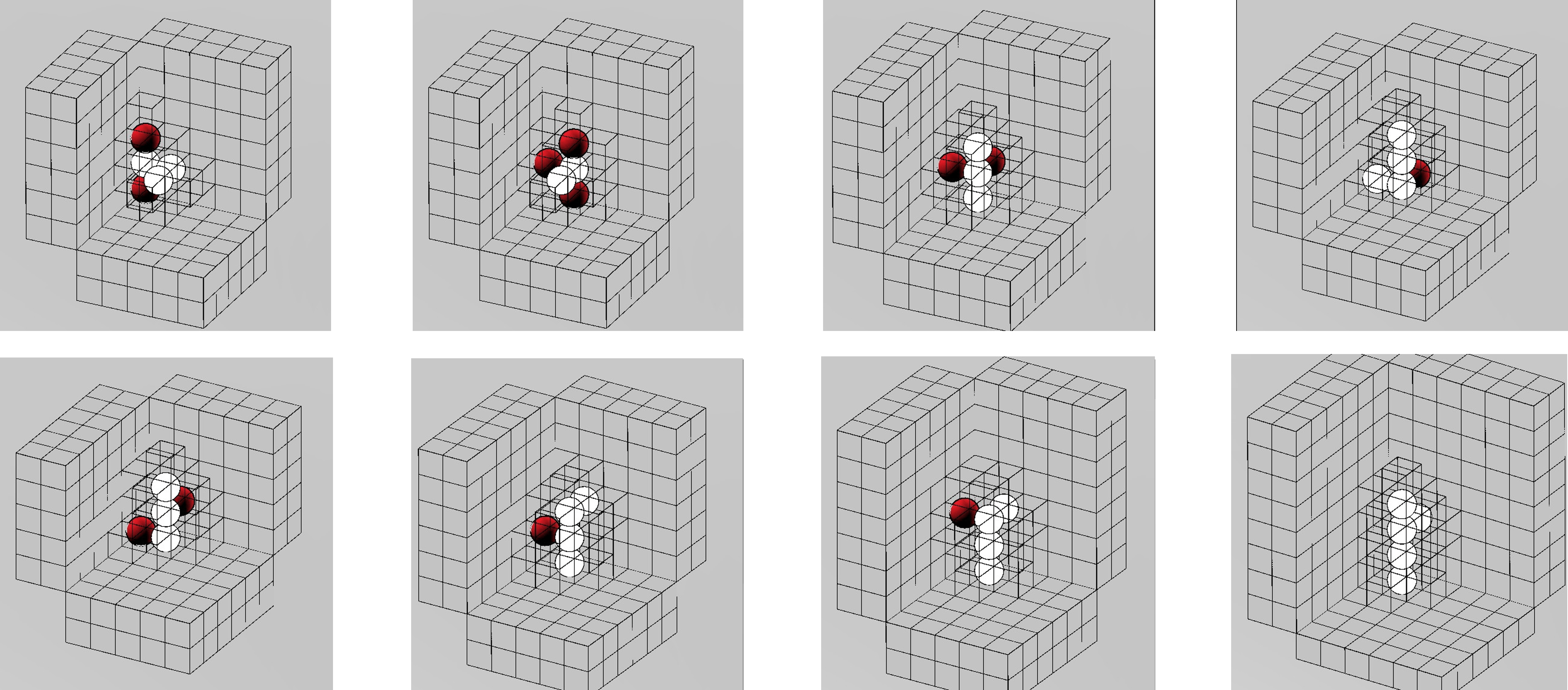}
             \caption{Move on the corner. In each figure, the red modules move.}
             \label{5MoveOnTheCorner}
            \end{figure}
     
        \begin{figure}[htbp]
             \centering
             \includegraphics[keepaspectratio, width=\linewidth]
                  {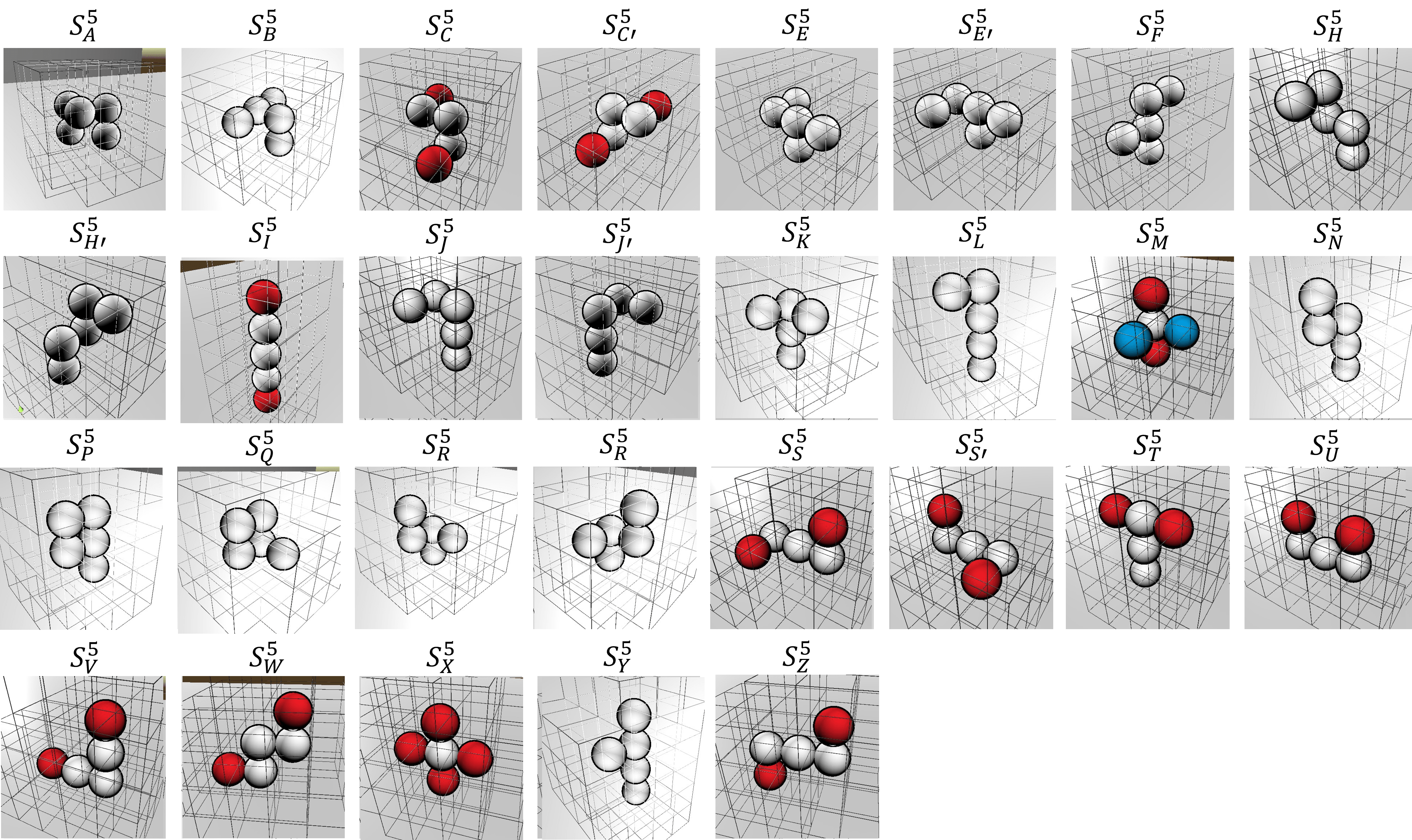}
             \caption{Configurations of the MRS of five modules not equipped with a common compass}
             \label{5ModuleShapes}
            \end{figure}
            
         \begin{figure}[htbp]
           \centering
           \includegraphics[keepaspectratio, scale=0.5]
                {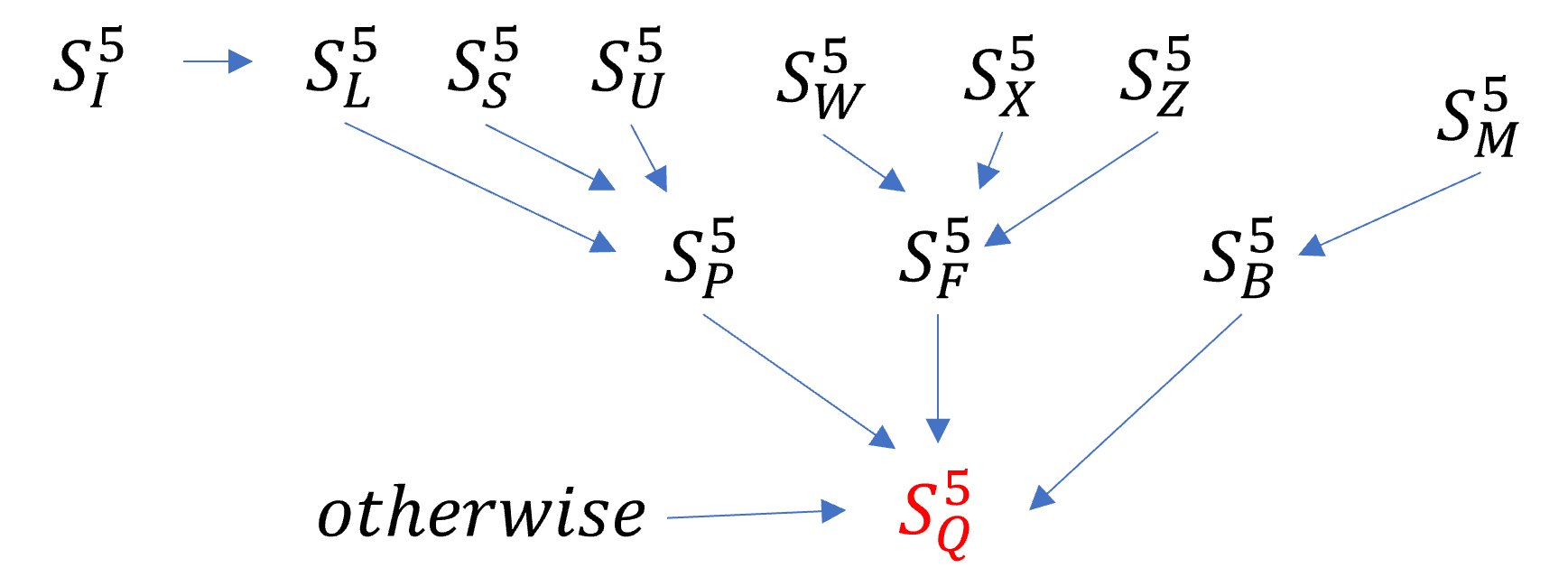}
           \caption{Transition graph of the MRS of five modules not equipped with a common compass}
           \label{5ModuleTransitionGraph}
          \end{figure}
            
           \begin{figure}[htbp]
             \centering
             \includegraphics[keepaspectratio, height=2cm]
                  {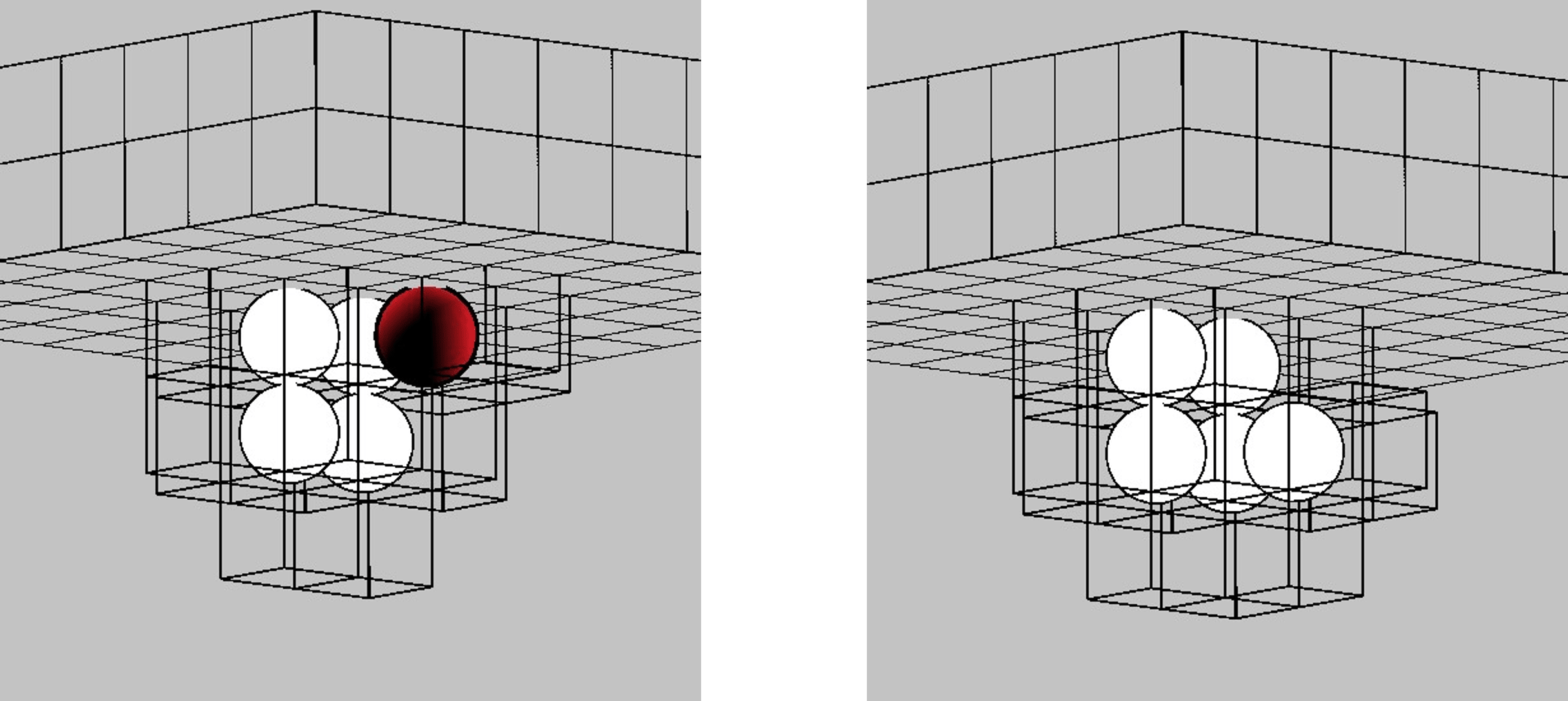}
             \caption{Transformation of $S^5_Q$}
             \label{Qtransform}
            \end{figure}
            
     Depending on the initial configuration, the search may not be possible, 
     and even if it is possible, the MRS cannot start the search directly because some initial configurations does not satisfy the 
     conditions in Table~\ref{tab_5modules}. 
     Figure~\ref{5ModuleShapes} shows all possible configurations for five modules not equipped with a common compass.
     The configurations containing the coloured modules are symmetric, i.e., 
     in each figure the modules painted by the same color may 
     have an identical observation and move symmetrically. 
     The MRS cannot start the search because it cannot move from their positions.
     Therefore, if there are pairs of modules with the same observation in an initial configuration, 
     the search is not possible.
     
     We show that the search is possible when all modules have different observations in the initial configuration. 
     Figure~\ref{5ModuleTransitionGraph} shows that the configurations $S^5_L, S^5_S, S^5_U$ can transform into $S^5_P$, $S^5_W, S^5_X, S^5_Z$ can transform into $S^5_F$,
     $S^5_M$ can transform into $S^5_B$, and the other configurations can transform into $S^5_Q$. 
     Even if some walls prevents these transformations, 
     the MRS can leave the wall because its state is asymmetric. 
     Then, any configuration can be transformed to $S^5_Q$.
     If the MRS satisfies the condition of the fourth or fifth movement of Table~\ref{tab_5modules}, 
     the MRS can start the search directly.
     Otherwise, the MRS cannot move because it touches a wall. 
     In this case, as shown in Figure~\ref{Qtransform}, the MRS in $S^5_Q$ on a wall 
     can change its direction by a single sliding 
     and in this new configuration it can perform one of these movements. 
     By adding the above rules to the algorithm, the MRS can start the search.

\begin{center}
     \begin{longtable}{|l|c|c|c|c|}
     
       \caption{Search algorithm for the MRS of five modules not equipped with a common compass}
     
     \endhead
     \endfoot
     \cline{1-5}          & \multicolumn{1}{c|}{$C_{m}$}        & \multicolumn{1}{c|}{$C_{w}$} & \multicolumn{1}{c|}{$C_e$} & \multicolumn{1}{c|}{Output} \\
     \cline{1-5}      $M_{Forward}$ & $(0,0,1)$,$(0,-1,1)$,&                 & $(0,0,4)$         & $(-1,0,1)$ \\
                                   & $(0,0,2)$,$(0,0,3)$&                       &          &               \\ 
     \cline{2-5}         & $(0,1,0)$,$(0,1,-1)$,&                 & $(0,0,3)$ &           $(0,0,1)$ \\
                          & $(0,1,1)$,$(0,1,2)$ &                    &       &  \\
     \cline{2-5}          & $(0,0,-1)$,$(0,0,-2)$,&              & $(0,0,1)$    & $(-1,0,-1)$ \\
                         & $(0,-1,-2)$,$(0,0,-3)$ &                  &       &  \\
     \cline{2-5}        & $(0,1,0)$,$(-1,1,0)$,&                   & $(0,0,1)$  & $(0,1,1)$ \\
                         & $(0,1,-1)$,$(-1,1,-1)$ &               &       &  \\
     \cline{2-5}        & $(1,0,0)$,$(1,-1,0)$,&                  & $(0,0,2)$       & $(0,-1,1)$ \\
                         & $(0,0,-1)$,$(1,0,-1)$ &                    &              &       \\
     \cline{2-5}        & $(0,1,0)$,$(1,1,0)$,&                     & $(0,0,2)$       & $(0,1,1)$ \\
                         & $(1,1,1)$,$(1,1,-1)$ &                    &          &    \\
     \cline{2-5}       & $(1,0,0)$,$(0,0,-1)$,&                      & $(0,0,1)$          & $(1,0,1)$ \\
                          & $(1,0,-1)$,$(1,0,-2)$ &                    &                 &       \\
     \cline{2-5}        & $(1,0,0)$,$(0,0,1)$ &                   & $(0,0,2)$ &          $(1,-1,0)$ \\
                          & $(1,0,1)$,$(1,0,-1)$,&                   &       &      \\
     \cline{1-5}     $M_{TurnB}$ & $(0,1,0)$,$(0,1,1)$,& $(0,0,3)$              & $(1,0,0)$ &            $(-1,1,0)$ \\
                                    & $(0,1,-1)$,$(0,1,2)$ &                      &           &      \\
     \cline{2-5}        & $(0,0,-1)$,$(0,0,-2)$ & $(0,0,1)$              & $(1,0,0)$         & $(0,-1,-1)$ \\
                         & $(0,-1,-2)$,$(0,0,-3)$ &                            &               &    \\
     \cline{1-5}    $M_{Back}$ & $(0,0,-1)$,$(-1,0,-1)$,&                 & $(0,0,1)$,& $(0,-1,-1)$ \\
                               & $(0,0,-2)$,$(1,0,-2)$ &                    &   $(0,0,-3)$    &   \\
     \cline{2-5}         & $(1,0,0)$,$(1,0,1)$,&                    & $(0,0,2)$          & $(0,0,-1)$ \\
                          & $(1,0,-1)$,$(2,0,-1)$&                       &                  &             \\
      \cline{2-5}        & $(-1,0,0)$,$(-1,0,1)$,&                      & $(0,0,3)$,& $(-1,0,-1)$ \\
                          & $(-2,0,1)$,$(-1,0,2)$&                    &   $(0,0,-1)$ &           \\
     \cline{2-5}       & $(0,1,0)$,$(0,1,-1)$,&                       & $(0,0,-3)$                & $(0,0,-1)$ \\
                          & $(-1,1,-1)$,$(0,1,-2)$ &                  &                        &  \\
     \cline{2-5}       & $(0,0,1)$,$(-1,0,1)$,&                     & $(0,0,-1)$          & $(-1,0,0)$ \\
                          & $(0,0,2)$,$(0,-1,2)$ &                       &           &  \\
     \cline{2-5}       & $(0,0,1)$,$(1,0,1)$,&                   & $(0,0,-1)$,& $(1,0,0)$ \\
                          & $(1,-1,1)$,$(1,0,2)$ &                    &   $(2,0,0)$         &   \\
     \cline{2-5}       & $(0,0,-1)$,$(-1,0,-1)$,&                    &             & $(-1,0,0)$ \\
                          & $(0,-1,-1)$,$(-1,0,-2)$ &                     &                &    \\
     \cline{2-5}     & $(0,0,-1)$,$(1,0,-1)$,&                &              & $(1,0,0)$ \\
                          & $(1,-1,-1)$,$(1,0,-2)$ &                            &       &       \\
     \cline{2-5}     & $(0,1,0)$,$(-1,1,0)$,&                    & $(0,0,-2)$       & $(0,0,-1)$ \\
                          & $(-1,1,1)$,$(0,1,-1)$ &                       &            &            \\
     \cline{2-5}     & $(0,1,0)$,$(0,1,1)$,&                  & $(0,0,-1)$,& $(1,1,0)$ \\
                         & $(-1,1,1)$,$(0,1,2)$ &                   &    $(0,0,3)$        &     \\
     \cline{1-5}   $M_{TurnF}$ & $(0,1,0)$,$(0,1,-1)$,& $(0,0,-3)$              &             & $(-1,1,0)$ \\
                                    & $(-1,1,-1)$,$(0,1,-2)$&                   &               &    \\
     \cline{2-5}         & $(1,0,0)$,$(1,0,1)$,& $(0,0,-2)$         &              & $(0,0,-1)$ \\
                              & $(1,-1,1)$,$(1,0,-1)$ &           &                       &  \\
     \cline{2-5}          & $(1,0,0)$,$(1,0,1)$,& $(0,0,-1)$                         &       & $(1,-1,0)$ \\
                          & $(1,0,2)$,$(0,0,2)$ &                     &                          &       \\
     \cline{2-5}          & $(0,1,0)$,$(0,1,1)$,& $(0,0,-1)$          &           & $(1,1,0)$ \\
                              & $(0,1,2)$,$(-1,1,2)$ &                    &                     &    \\
     \cline{2-5}         & $(-1,0,0)$,$(-1,0,1)$,& $(0,0,-1)$    &                 & $(0,0,1)$ \\
                          & $(-1,0,2)$,$(-2,0,2)$ &                            &             &   \\
     \cline{2-5}          & $(0,0,1)$,$(1,0,1)$,& $(0,0,-1)$                &       & $(1,0,0)$ \\
                          & $(0,0,2)$,$(-1,0,2)$,&                   &              &  \\
     \cline{2-5}          & $(1,0,0)$,$(1,0,-1)$    & $(0,0,-3)$       &              & $(1,-1,0)$ \\
                          & $(2,0,-1)$,$(1,0,-2)$ &                 &                  &         \\
     \cline{2-5}          & $(0,1,0)$,$(0,1,-1)$,& $(0,0,-3)$           &       & $(1,1,0)$ \\
                               & $(1,1,-1)$,$(1,1,-2)$ &              &       &        \\
     \cline{1-5}    $M_{Edge}$ & $(0,0,1)$,$(0,-1,1)$,& $(0,0,4)$,$(1,0,0)$ &              & $(-1,0,1)$ \\
                                    & $(0,0,2)$,$(0,0,3)$ &           &       &        \\
     \cline{2-5}          & $(0,1,0)$,$(0,1,1)$,& $(0,0,3)$,$(1,0,0)$ &              & $(0,0,1)$ \\
                          & $(0,1,2)$,$(0,1,-1)$ &                    &       &               \\
     \cline{2-5}         & $(0,0,-1)$,$(0,0,-2)$,& $(0,0,1)$,$(1,0,0)$ &             & $(-1,0,-1)$ \\
                               & $(0,-1,-2)$,$(0,0,-3)$ &                 &       &         \\
     \cline{2-5}         & $(0,1,0)$,$(1,1,0)$,& $(2,0,0)$,$(0,0,2)$ &         $(0,2,0)$& $(0,1,1)$ \\
                              & $(1,1,1)$,$(1,1,-1)$&                &           &    \\
     \cline{2-5}          & $(1,0,0)$,$(0,0,-1)$,& $(0,0,1)$,$(2,0,0)$ &                       & $(0,1,-1)$ \\
                          & $(1,0,-1)$,$(1,0,-2)$ &                   &              &     \\
     \cline{2-5}          & $(0,0,1)$,$(-1,0,1)$,& $(0,0,3)$,$(1,0,0)$ &              & $(0,1,1)$ \\
                              & $(0,0,2)$,$(-1,0,2)$&                &               &     \\
     \cline{2-5}          & $(0,0,-1)$,$(0,1,-1)$,& $(0,0,1)$,$(1,0,0)$ &               & $(0,1,0)$ \\
                          & $(-1,0,-1)$,$(-1,1,-1)$ &               &               &          \\
     \cline{2-5}         & $(1,0,0)$,$(0,1,0)$,& $(0,0,2)$,$(2,0,0)$ &           & $(-1,1,0)$ \\
                              & $(1,1,0)$,$(1,0,1)$ &                 &              &     \\
     \cline{2-5}          & $(0,1,0)$,$(-1,1,0)$,& $(0,0,2)$,$(1,0,0)$ &              & $(-1,0,0)$ \\
                          & $(-2,1,0)$,$(0,1,1)$ &                    &                &        \\
     \cline{2-5}          & $(0,1,0)$,$(1,1,0)$,& $(0,0,2)$,$(2,0,0)$ &              & $(-1,0,0)$ \\
                          & $(-1,1,0)$,$(1,1,1)$ &                    &              &     \\
     \cline{2-5}          & $(0,0,-1)$,$(-1,0,-1)$,& $(0,0,1)$,$(1,0,0)$ &             & $(-1,0,0)$ \\
                              & $(-2,0,-1)$,$(-1,-1,-1)$ &                &           &     \\
     \cline{2-5}          & $(0,0,-1)$,$(1,0,-1)$,& $(0,0,1)$,$(2,0,0)$ &                & $(-1,0,0)$ \\
                          & $(-1,0,-1)$,$(-1,-1,-1)$ &              &             &         \\
     \cline{2-5}         & $(0,0,-1)$,$(1,0,-1)$,& $(0,0,1)$,$(3,0,0)$ &             & $(-1,0,-1)$ \\
                          & $(2,0,-1)$,$(0,-1,-1)$ &                      &          &  \\
     \cline{1-5}   $M_{Corner}$ & $(0,0,-1)$,$(-1,0,-1)$,& $(0,0,1)$,$(1,0,0)$,&            & $(-1,0,0)$ \\
                                         & $(-1,-1,-1)$,$(0,0,-2)$ &  $(0,1,0)$              &              &    \\
     \cline{2-5}          & $(0,0,1)$,$(-1,0,1)$,& $(0,0,3)$,$(1,0,0)$,&              & $(-1,0,0)$ \\
                          & $(-1,-1,1)$,$(0,0,2)$ &    $(0,1,0)$           &           &     \\
     \cline{2-5}         & $(0,0,1)$,$(1,0,1)$,& $(0,0,3)$,$(2,0,0)$,&              & $(0,-1,0)$ \\
                          & $(0,-1,1)$,$(0,0,2)$&     $(0,1,0)$              &               &    \\
     \cline{2-5}         & $(-1,0,0)$,$(-1,-1,0)$,& $(0,0,2)$,$(1,0,0)$,&            & $(0,-1,0)$ \\
                          & $(-1,0,-1)$,$(-1,0,1)$ &  $(0,1,0)$          &         &    \\
     \cline{2-5}         & $(0,0,-1)$,$(1,0,-1)$,& $(0,0,1)$,$(2,0,0)$,&             & $(0,-1,0)$ \\
                         & $(0,-1,-1)$,$(0,0,-2)$ & $(0,1,0)$     &         &     \\
     \cline{2-5}        & $(-1,0,0)$,$(-1,-1,0)$,& $(0,0,2)$,$(1,0,0)$,&             & $(0,0,1)$ \\
                         & $(-1,0,-1)$,$(-1,0,1)$ &  $(0,-2,0)$              &               &    \\
     \cline{2-5}        & $(-1,0,0)$,$(-1,1,0)$,& $(0,0,2)$,$(1,0,0)$,&             & $(0,0,1)$ \\
                         & $(-1,0,1)$,$(-1,0,-1)$ &   $(0,2,0)$          &           &     \\
               \cline{1-5}          & \multicolumn{1}{c|}{Otherwise} & \multicolumn{1}{c|}{Otherwise} & \multicolumn{1}{c|}{Otherwise} & $(0,0,0)$ \\
               \cline{1-5}  
     \end{longtable}%
     \end{center} 
     \label{tab_5modules}%

\section{Necessary number of modules} 

We show that the three algorithms presented in 
Section~\ref{ExpWithCompass} uses the 
minimum number of modules for each settings. 

\begin{theorem}
\label{theorem: 2 modules cant Exprole 3D space with compass}
The MRS of less than three modules equipped with a common compass cannot solve the search problem 
in a finite 3D cubic grid.
\end{theorem}

\begin{figure}[htbp]
     \centering
     \includegraphics[keepaspectratio,scale=0.35]
          {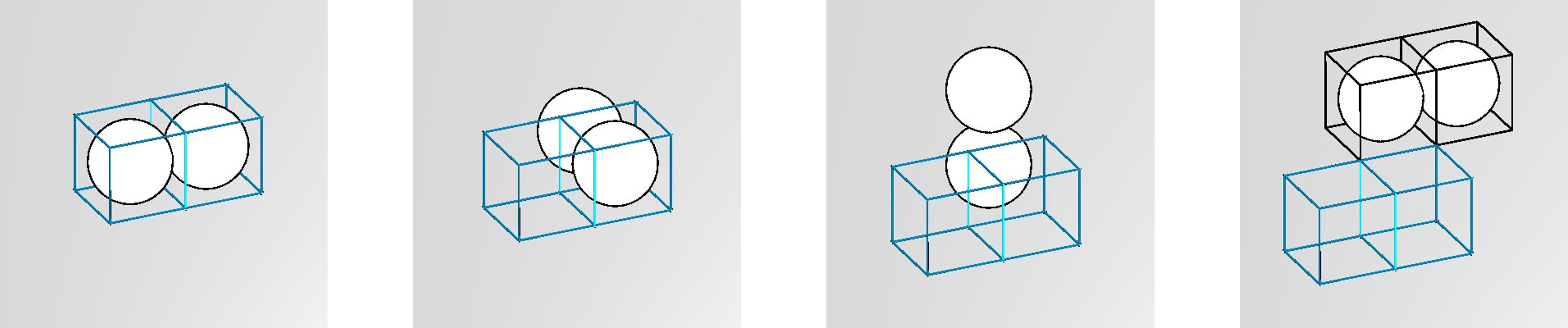}
     \caption{Example of diagonal move}
     \label{DiagonalMoveExample}
    \end{figure}
   \begin{figure}[htbp]
     \centering
     \includegraphics[keepaspectratio, scale=0.25]
          {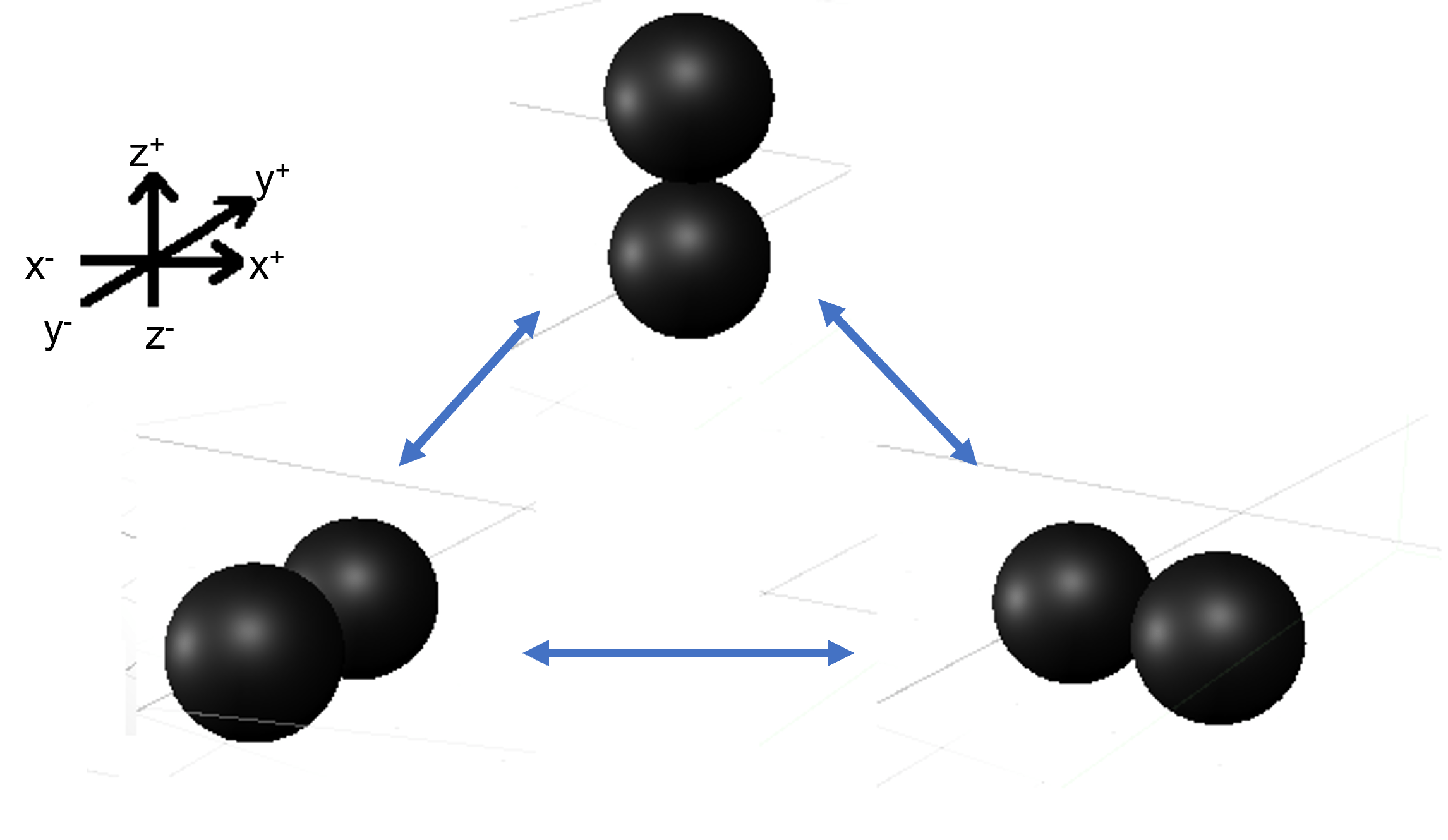}
     \caption{Configurations of two modules with a common compass}
     \label{2ModulesShapes}
    \end{figure}

\begin{proof}
In the case of one module, the MRS cannot move because there is no other static module during sliding or rotation. 

In the case of two modules, we show that 
there is a cell which cannot be visited by the MRS. 
Doi et al. showed that in the 2D square grid two modules equipped with a common compass 
can move straight to one of the eight directions 
(north, south, east, west, northeast, northwest, southeast, and southwest) 
when they observe no wall~\cite{DYKY18}.
By the same discussion, the MRS of two modules equipped with a common compass 
in 3D cubic grid can move straight to one direction, and 
the possible moving directions are eight diagonal directions 
in addition to those eight directions on planes perpendicular 
to $x$, $y$, or $z$ axis 
when they can observe no wall. 
Figure~\ref{DiagonalMoveExample} shows an example of a 3D diagonal move.  

Without loss of generality, we assume that the MRS continues to move in positive $x$ direction 
in the area where the two modules observe no wall. 
If it enters the area from other walls than the west wall, 
it cannot visit the cells with smaller $x$ coordinates on the straight moving track. 
Hence, the MRS must start the straight move when some modules can observe the west wall. 

Assume that the MRS can visit all cells of the field. 
There exists a subsequence of configurations where the two modules observes only the west wall. 
We focus on the sequence of states of the MRS. 
Let $T_W = S_1, S_2, \ldots, S_{\ell}$ be a sequence of states of the MRS 
in the area where the MRS observes only west wall.
$T_W$ consists of different states, otherwise the MRS cannot leave the west wall. 
Additionally, in the transformation from $S_{\ell}$ to the next state, say $S_{\ell+1}$ 
the MRS enters the area where it observes no wall. 
As shown in Figure~\ref{2ModulesShapes}, when we fix the bottom westmost module, 
there are three states of the MRS. 
The number of states where the modules observe only the west wall is at most $3k$. 
Thus, the length of $T_W$ is at most $3k$. 
In addition, the maximum distance that the MRS can move by $T_W$ in the $y$ direction 
or the $z$ direction is $3k$, 
since the maximum number of coordinates that it can travel by one transition is one in each direction.

Since the MRS cannot move in the negative $x$ direction from the area 
where the modules cannot observe the west wall, 
the only way for the MRS to enter the area where the modules can observe the west wall 
is to move from the area where the modules observe two or more walls. 
Thus, the transformation sequence $T_W$ always starts around the edge of the field that touches the west wall.
However, since the maximum distance that the MRS can travel by $T_W$ is $3k$, 
it cannot reach the coordinates $(a,b,c)$ such that $(0 <= a <= k, 3k+k < b < d-(3k+k), 3k+k < c < h-(3k+k))$.
Hence, it is impossible for the two modules to visit all cells.
\end{proof}

We next show the necessary number of modules equipped with a common vertical axis. 
\begin{theorem}
\label{theorem: 3 modules cant Exprole 3D space with horizontal compass}
The MRS of less than four modules equipped with a common vertical axis 
cannot solve the search problem in a finite 3D cubic grid.
\end{theorem}
\begin{proof} 
\begin{figure}[htbp]
  \centering
  \includegraphics[keepaspectratio, scale=0.23]
       {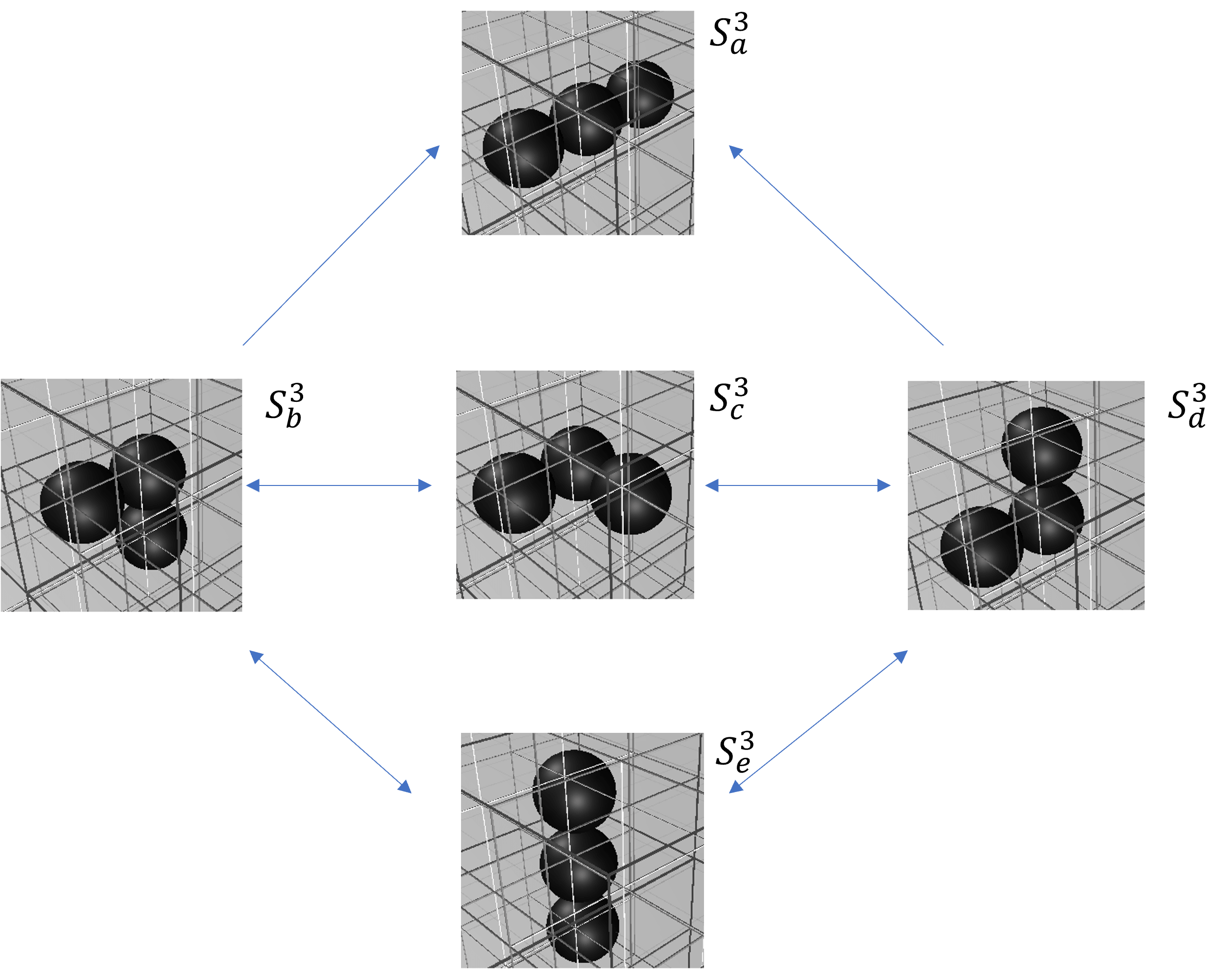}
  \caption{State transition graph for a MRS consisting of 3 modules with 
  common vertical axis}
  \label{3moduleCantMoveHorizon}
 \end{figure}
  
In the case of one module, the MRS cannot move because it cannot perform any sliding or rotation.

In the case of two modules, there are two possible states of the MRS. 
Let $S_A$ be the state, where the two modules form a vertical line, 
and $S_B$ be the state, where the two modules form a horizontal line. 
There exists only one horizontal line state because 
the modules do not agree on $x$ axis or $y$ axis. 
In $S_A$, one of the two modules can perform a rotation 
because they agree on a common vertical axis. 
Any rotation in $S_A$ results in $S_B$. 
In $S_B$, both modules obtain the same observation 
if their local coordinate systems are symmetric against their midpoint, 
and if one of them moves then the other also moves. 
Thus, the two modules cannot perform any movement. 
Consequently, the two modules cannot move to any direction. 

In the case of three modules, we check possible movements of the MRS 
by the state transition graph shown in Figure~\ref{3moduleCantMoveHorizon}.
State $S^3_a$ cannot be transformed to any other configuration 
because both 
endpoint modules get the same observation.
Therefore, it is necessary to move only in the $S^3_b, S^3_c, S^3_d$, and $S^3_e$.
However, no matter which transformation of the $S^3_b, S^3_c, S^3_d$, and $S^3_e$,
the MRS cannot move in the east, west, south or south direction.
Therefore, when there is no wall in the vicinity, it becomes impossible to move,
and it is impossible to search the whole field.
\end{proof} 

We finally show the necessary number of modules not equipped with a common compass. 
\begin{theorem}
\label{theorem: 4 modules cant Exprole 3D space without compass }
The MRS of less than five modules not equipped with a common compass 
cannot solve the search problem in a finite 3D cubic grid.
\end{theorem}

\begin{proof} 
     \begin{figure}[htbp]
     \centering
     \includegraphics[keepaspectratio,scale=0.25]
        {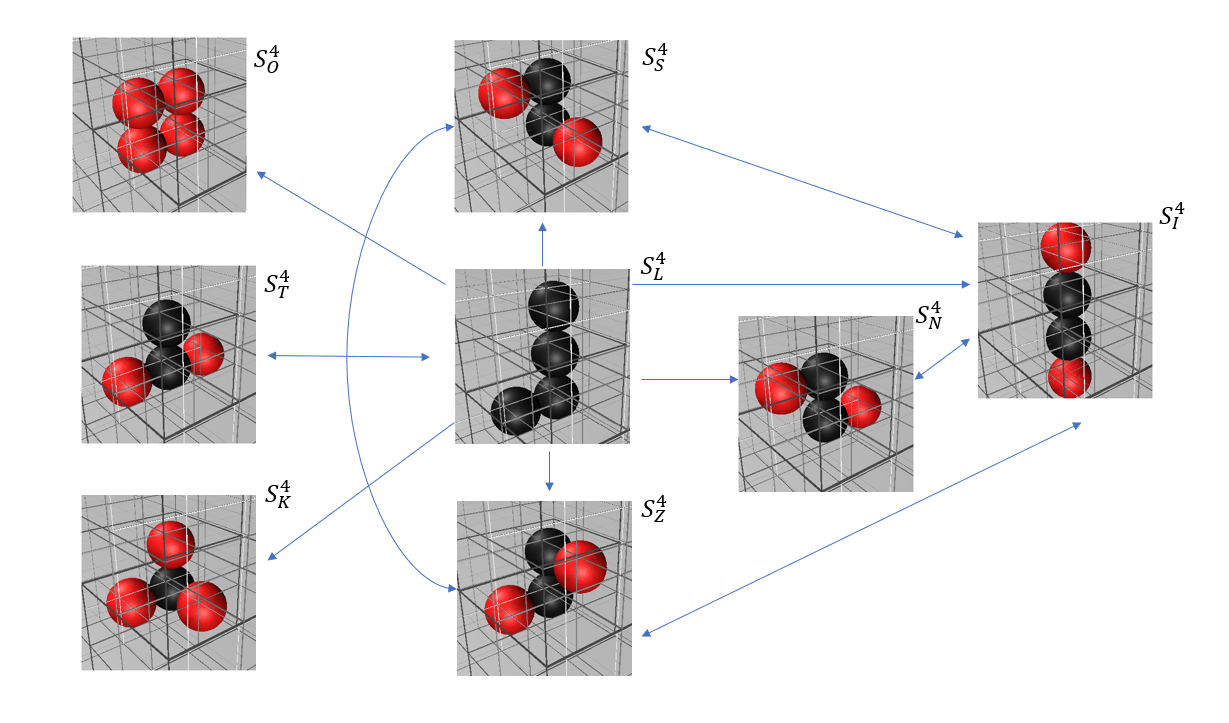}
     \caption{State transition graph for the MRS of 4 modules not equipped with a common compass}
     \label{4moduleCantMoveNoCompass}
     \end{figure}
     
     In the case of one module, the MRS cannot move because it cannot perform any sliding or rotation.
     
     In the case of two modules, the MRS can take one state, 
     where the two modules obtain the same observation 
     if their local coordinate systems are symmetric against their midpoint. 
     Thus, if one of them moves then the other also moves and 
     the two modules cannot perform any movement. 
     
     In the case of three modules, there are two possible states of the MRS, i.e., 
     the L-shape shown in the middle of the Figure~\ref{3moduleCantMoveHorizon}, 
     and the I-shape, i.e., a line. 
     In the L-shape, the two endpoint modules obtain the same observation 
     if their local coordinate systems are symmetric against the center module. 
     The center module cannot perform any movement due to connectivity. 
     The two endpoint modules cannot perform a sliding because they move into the same cell 
     and there is no static modules during a sliding. 
     If the two modules perform a rotation, the resulting state is the L-shape 
     or the I-shape. 
     In the I-shape, the two endpoint modules obtain the same observation in the worst case. 
     They can only perform a rotation, and a resulting configuration 
     is the L-shape or the I-shape. 
     In any movement in the L-shape and the I-shape, 
     the center module does not move and the MRS does not move to any direction. 
     Consequently, the MRS cannot move to any direction. 
     
     In the case of four modules, the state transition graph is shown in Figure~\ref{4moduleCantMoveNoCompass}. 
     In each state, modules that may have the same observation is painted in red, 
     and each arrow shows that its starting point state can be transformed to 
     its end point in one time step. 
     The MRS cannot change its state in $S_O^4$ and $S_K^4$. 
     In $S_O^4$ the four modules have the same observation, and 
     if one of them moves, the other also move. 
     Thus, they cannot keep a static module during movement. 
     Also in $S_K^4$ the three red module have the same observation, and
     if one of them moves, the others also move.
     It cause a conflict of their movement paths or transformation into $S_K^4$.
     Thus, they can transform only into $S_K^4$.
     
     Next, we can find $S_S^4$, $S_Z^4$, and $S_I^4$ form a loop. 
     However, the MRS cannot move to any direction by repeating these transitions. 
     The MRS enters this loop from $S_N^4$, thus 
     it cannot move to any direction from $S_N^4$. 
     We also find $S_T^4$ and $S_L^4$ forms a loop and 
     $S_L^4$ has arrows to $S_O^4$, $S_K^4$, $S_I^4$, and $S_N^4$. 
     Thus, the MRS cannot move to any direction from $S_T^4$ and $S_L^4$. 
     Therefore, there is no way for MRS to move.
     
     \end{proof}

\section{Conclusion and future work}
\label{conc}

In this paper, we considered search by the single MRS 
in the finite 3D cubic grid. 
We demonstrated a trade-off between the common knowledge and 
the necessary and sufficient number of modules for search. 
When the modules are equipped with a common compass, 
three modules are necessary and sufficient. 
When the modules are not equipped with a common compass, 
five modules are necessary and sufficient. 
As an intermediate case, when the modules are equipped with 
a common vertical axis, four modules are necessary and sufficient. 
We finally note that the proposed algorithms depend on 
parallel movements, i.e., they are not designed for 
the centralized scheduler. 

Our future goal is a distributed coordination theory for the MRS.  
First, reconfiguration and locomotion 
of a single MRS in the 3D cubic grid have not been discussed yet. 
Second, it is important to consider interaction among multiple MRSs 
such as rendezvous, collision avoidance, and collective search. 
Finally, the MRS is expected to solve more complicated tasks 
by interaction with the environment.


\newpage 

\begin{thebibliography}{99}

\bibitem{AMP20a}
Abdullah Almethen, Othon Michail, and Igor Potapov. 
Pushing lines helps: Efficient universal centralised transformations for programmable matter. 
{\it Theoretical Computer Science,} 
830-831, pp.43--59, 2020.  

\bibitem{AMP20b}
Abdullah Almethen, Othon Michail, and Igor Potapov. 
On Efficient Connectivity-Preserving Transformations in a Grid. 
{\it In Proc. of ALGOSENSORS 2020}, pp.76--91, 2020. 

\bibitem{AADFP04}
Dana Angluin, James Aspnes, Zo\"{e} Diamadi, Michael J. Fischer, Ren\'{e} Peralta, 
Computation in networks of passively mobile finite-state sensors. 
{\it In Proc. of PODC 2004}, pp.290--299, 2004. 


\bibitem{DDGRSS14} 
Zahra Derakhshandeh, Shlomi Dolev, Robert Gmyr, Andr\'{e}a W. Richa, Christian Scheideler, and Thim Strothmann. 
Brief announcement: amoebot - a new model for programmable matter. 
{\it In Proc. of SPAA 2014}, pp.220--222, 2014. 

\bibitem{DGRSS15} 
Zahra Derakhshandeh, Robert Gmyr, Andr\'{e}a W. Richa, Christian Scheideler, and Thim Strothmann, 
An Algorithmic Framework for Shape Formation Problems in Self-Organizing Particle Systems. 
{\it In Proc. of NANOCOM 2015}, pp.21:1--21:2, 2015. 

\bibitem{DFSVY20}
Giuseppe Antonio Di Luna, Paola Flocchini, Nicola Santoro, Giovanni Viglietta, and Yukiko Yamauchi. 
Shape formation by programmable particles. 
{\it Distributed Computing,} 33(1), pp.69--101, 2020. 

\bibitem{DYKY18} 
Keisuke Doi, Yukiko Yamauchi, Shuji Kijima, and Masafumi Yamashita. 
Exploration of Finite 2D Square Grid by a Metamorphic Robotic System. 
{\it In Proc. of SSS 2018,} pp.96--110, 2018. 

\bibitem{DP06} 
Adrian Dumitrescu and J\'{a}nos Pach. 
Pushing Squares Around. 
{\it Graphs and Combinatorics}, 
22, pp.37--50, 2006. 

\bibitem{DSY04a} 
Adrian Dumitrescu, Ichiro Suzuki, and Masafumi Yamashita. 
Motion planning for metamorphic systems: feasibility, decidability, and distributed reconfiguration. 
{\it IEEE Transactions on Robotics}, 
20(3), pp.409--418, 2004. 
 
\bibitem{DSY04b} 
Adrian Dumitrescu, Ichiro Suzuki, and Masafumi Yamashita. 
Formations for Fast Locomotion of Metamorphic Robotic Systems. 
{\it The International Journal of Robotics Research}, 
23(6), pp.583--593, 2004.

\bibitem{FYOKY15} 
Nao Fujinaga, Yukiko Yamauchi, Hirotaka Ono, Shuji Kijima, and Masafumi Yamashita. 
Pattern Formation by Oblivious Asynchronous Mobile Robots. 
{\it SIAM Journal on Computing}, 44(3), pp.740--785, 2015. 

\bibitem{MSS19}
Othon Michail, George Skretas, and Paul G.~Spirakis, 
On the transformation capability of feasible mechanisms for programmable matter. 
{\it Journal of Computer and System Sciences,} 102, pp.18--39, 2019. 

\bibitem{NSY20}
Junya Nakamura, Sayaka Kamei and Yukiko Yamauchi.
Evacuation from a Finite 2D Square Grid Field
by a Metamorphic Robotic System.
{\it In Proc. of CANDAR 2020}, pp.69--78,
2020.

\bibitem{RCN14}
Michael Rubenstein, Alejandro Cornejo, and Radhika Nagpal. 
Programmable self-assembly in a thousand-robot swarm. 
{\it Science}, 345(6198), pp.795--799, 2014. 

\bibitem{SY99}
Ichiro Suzuki and Masafumi Yamashita, 
Distributed Anonymous Mobile Robots: Formation of Geometric Patterns. 
{\it SIAM Journal on Computing,} 28(4), pp.1347--1363, 1999. 

\bibitem{TPB19} 
Pierre Thalamy, Beno\^{i}t Piranda, and Julien Bourgeois, 
Distributed Self-Reconfiguration using a Deterministic Autonomous Scaffolding Structure. 
{\it In Proc. of AAMAS 2019}, pp.140--148, 2019. 

\bibitem{YY21} 
http://tcs.inf.kyushu-u.ac.jp/~yamauchi/MRSdemonstrations.html

\end{thebibliography}
\end{document}